\documentclass[preprint,review]{elsarticle}

\usepackage[utf8]{inputenc}
\usepackage[T1]{fontenc}
\usepackage[french,english]{babel}
\usepackage{comment}
\keywordtitle{Keywords}

\usepackage{nicefrac}

\usepackage{bm}

\usepackage{amsmath,amssymb,amsfonts}

\usepackage[ruled,linesnumbered]{algorithm2e}

\usepackage{algorithmic}
\usepackage{tikz}
\usetikzlibrary{shapes.arrows}

\usepackage{graphicx}
\usepackage{textcomp}
\usepackage{xcolor}

\usepackage{mathtools, stmaryrd}
\usepackage{xparse} 
\usepackage{mathtools, nccmath,amsmath,amssymb,amsfonts,amstext, dsfont, color}
\usepackage{subcaption}
\usepackage{float}
\usepackage{lineno}

\usepackage{array,multirow}

\newtheorem{theorem}{Theorem}
\newtheorem{proposition}{Proposition}
\newtheorem{lemma}{Lemma}

\newtheorem{remark}{Remark}
\newtheorem{corollary}{Corollary}

\newtheorem{assumption}{Assumption}
\newtheorem{proof}{Proof}

\newcommand{\norm}[1]{\left\lVert #1 \right\rVert}

\newcommand\myon{ \left(\frac{1}{\lambda_n^i}\cdot \frac{\lambda_n^r } {\e^{\lambda_n^r t_C}-1}  + 1 \right)^{-1}}

\newcommand{\myonn}[3]{ \left(\frac{1}{#2}\cdot \frac{#1} {\e^{#1 #3}-1}  + 1 \right)^{-1}}

\newcommand\hit{h}

\newcommand\occupancy{o}
\newcommand\newoccupancy{o}

\newcommand\givenbase[1][]{\:#1\lvert\:}
\let\given\givenbase

\DeclarePairedDelimiterX\Basics[1](){\let\given\sgiven #1}

\newcommand{\Proba}[1]{\mathrm{Pr}\left(#1\right)}
\newcommand{\E}[1]{\mathbb{E}\left[#1\right]}
 
\newcommand{\Ton}{T^{\mathrm{On}}_{n} }
\newcommand{\Toff}{T_{n}^{\mathrm{Off}} }
\newcommand{\Nopen}[1]{\mathcal{N}(#1)}
\newcommand{\Nclosed}[1]{ \mathcal{N}[#1]}
\newcommand{\Nopenin}[2]{\mathcal{N}_{#1}(#2)}
\newcommand{\Nclosedin}[2]{\mathcal{N}_{#1}[#2]}

\newcommand{\e}{\mathrm{e}}

\newcommand{\foccupancyC}{g}
\newcommand{\foccupancy}{\mathbf{g}}
\newcommand{\frefresh}{\mathbf{R}}
\newcommand{\fentry}{\mathbf{E}}
\newcommand{\ftc}{t_C} 

\newcommand{\JacobianG}{\mathcal{J}_{\mathbf{G}}}
\newcommand{\JacobianGB}{\mathcal{J}_{\mathbf{G}_{\beta}}}

\DeclarePairedDelimiterX{\Iintv}[1]{\llbracket}{\rrbracket}{\iintvargs{#1}}
\NewDocumentCommand{\iintvargs}{>{\SplitArgument{1}{,}}m}
{\iintvargsaux#1} %
\NewDocumentCommand{\iintvargsaux}{mm} 
{#1\mkern1.5mu..\mkern1.5mu#2}

\DeclareMathOperator*{\argmax}{arg\,max}
\DeclareMathOperator*{\argmin}{arg\,min}

\def\BibTeX{{\rm B\kern-.05em{\sc i\kern-.025em b}\kern-.08em
    T\kern-.1667em\lower.7ex\hbox{E}\kern-.125emX}}

\journal{Computer Networks}

\begin{document}
\pagenumbering{arabic}
\setcounter{page}{1}
\begin{frontmatter}

\title{Performance Model for Similarity Caching}

\author[inst1]{Younes Ben Mazziane}

\author[inst1]{Sara Alouf}

\author[inst1]{Giovanni Neglia}

\affiliation[inst1]{organization={Université Côte d'Azur, Inria},
            city={Sophia Antipolis},
            country={France}}

\author[inst2]{Daniel~S.~Menasche}

\affiliation[inst2]{organization={Federal University of Rio de Janeiro, UFRJ},
            city={Rio de Janeiro},
            country={Brazil}}

\begin{abstract}
Similarity caching allows requests for an item 
 to be served by a similar item. Applications include recommendation systems,
multimedia retrieval, and machine learning. Recently, many similarity
caching policies have been proposed, like SIM-LRU and
RND-LRU, but the performance analysis of their hit rate is still wanting.
In this paper, we show 
how to extend the popular time-to-live approximation in classic caching to similarity
caching. In particular, we propose a method to estimate the hit rate of the similarity caching policy RND-LRU. Our method, the RND-TTL approximation, introduces the RND-TTL cache model and then tunes its parameters in such a way to mimic the behavior of RND-LRU. The parameter tuning involves solving a fixed point system of equations for which we provide an algorithm for numerical resolution and sufficient conditions for its convergence. Our approach for approximating the hit rate of RND-LRU is evaluated on both synthetic and real world traces. 
\end{abstract}


\end{frontmatter}

\section{Introduction}

Many applications require to retrieve items similar to a given user's request. For example, in content-based image retrieval~\cite{falchi2008metric} systems, users can submit an image to obtain other visually similar images. A similarity cache may intercept the user's request, perform a local similarity search over the set of locally stored items, and then, if the search result is evaluated satisfactorily, provide it to the user. The cache thus may speed up the reply and reduce the load on the server, at the cost of providing items \emph{less similar} than those the server would provide.

Originally proposed for content-based image retrieval~\cite{falchi2008metric} and contextual advertising~\cite{pandey2009nearest}, similarity caches are now a building block for a large variety of machine learning based inference systems for recommendations~\cite{sermpezis18}, image recognition~\cite{drolia2017precog,venugopal2018shadow} and network traffic~\cite{dlcaching22} classification. In these cases, the similarity cache stores past queries and the respective inference results to serve future similar requests.
Motivated by the large number of applications, much effort has been devoted recently to formalize similarity caching~\cite{neglia2021similarity,garetto2021content} as well as to propose new caching policies~\cite{zhou2020adaptive,sabnis2021grades,salem2021accai}.

RND-LRU is a similarity caching policy proposed in the seminal paper \cite{pandey2009nearest}. It is a variant of the least recently used (LRU) policy adapted to the similarity caching setting. We still lack an analytical evaluation of RND-LRU. The aim of this paper is to fill this gap.


Even for the classic LRU policy, computing the hit rate, i.e., the fraction of requests satisfied by the cache, under the Independent Reference Model (IRM) \cite{irm-fagin-1977}, is a challenging task whose computational cost is exponential in both the cache size and the number of items \cite{king1972analysis,dan1990approximate}. 
The so called Che's or time-to-live (TTL) is a highly efficient method for accurately estimating the hit rate of LRU under IRM \cite{che2002hierarchical,fofack2014performance}.
The TTL approximation leverages the analysis of an opportune cache---which benefits from decoupling caching decisions across items---and utilizes its hit rate  
as an estimate for the hit rate of LRU. The TTL approximation has received theoretical support under different assumptions regarding the request process \cite{fricker2012versatile,jiang2018convergence,leonardi2015least}.

As items in a RND-LRU cache are strongly coupled, RND-LRU analysis is even more challenging.

In fact, in classic caching, an item in the cache only serves requests for itself, while in similarity caching, the cached item can serve requests for a set of similar items as long as neither these, nor their most similar items, are stored in the cache. It follows that, in similarity caching, the number of requests satisfied by a cached item depends in general on \emph{the whole cache state}.

In this paper, we extend the TTL approximation to RND-LRU, by introducing the RND-TTL approximation; the latter is based on a novel similarity caching model, that we call RND-TTL. This approximation involves tuning the parameters of the RND-TTL model to estimate the hit rate of RND-LRU. Our contributions can be summarized as follows:
\begin{itemize}
    \item We propose a novel similarity caching model named RND-TTL and we compute its the hit rate under IRM. 

    \item We establish conditions on the parameters of RND-TTL to ensure that it emulates the behavior of RND-LRU.

    \item The parameter tuning process for the RND-TTL model involves solving a system of fixed point equations; we present a parameterized iterative algorithm to solve this system and provide a practical method for selecting the algorithm's parameter.

    \item We provide sufficient conditions for the iterative algorithm to converge. 
    
    \item We evaluate the accuracy of our RND-TTL approximation to estimate the hit rate of RND-LRU on both synthetic and real-world traces.

\end{itemize}

This paper revisits and extends our previous work \cite{mazziane2022computing}.  In particular, we have reframed the analysis of the RND-LRU cache after introducing the RND-TTL approximation. Also, the convergence results of the iterative algorithm are new.

The rest of the paper is organized as follows: Sec.~\ref{s:background-notation} introduces background, notation and assumptions. Section~\ref{s:TTL-SIM-LRU} presents the RND-TTL approximation and Sec.~\ref{s:Algorithm-fixed-point} focuses on our iterative algorithm for tuning the parameters of the RND-TTL cache. We evaluate its performance in Sec.~\ref{s:experiments} on both synthetic and real word traces and summarize our findings in Sec.~\ref{Conclusion}.

\begin{table}
  \caption{Table of notation} \centering
\begin{tabular}{ |l|l| } 
 \multicolumn{2}{c}{{\bf Basic parameters}} \\
 \hline
 $\mathcal{I}$ & set of items \\
 $N=|\mathcal{I}|$ & catalog size \\
 $C$ & cache capacity \\
$S$ & state of cache; set of cached items \\
$\lambda_n$ & arrival rate of requests for item $n$ \\
$\mathrm{dis}(\cdot, \cdot)$ &function measuring the dissimilarity between items \\
 \hline
 \multicolumn{2}{c}{{\bf RND-LRU}} \\
\hline
$\tilde{\Omega}$ & RND-LRU cache state space \\ 
 $d$ & similarity threshold \\
  $\Nopen{n}$ & neighbors of item $n$ \\ 
 $\Nclosed{n}$&  neighbors of item $n$ including $n$ \\
 $\Nopenin{m}{n}$ & items in $\Nopen{n}$  strictly closer to $n$ than $m$ \\
 $\Nclosedin{m}{n}$ & items in $\Nclosed{n}$ strictly closer to $n$ than $m$ \\
 $q_n(m)$& probability to use candidate $n$ to serve a request for $m$\\ 
 $\tilde{\lambda}_n^i$& insertion rate of item $n$ \\
 $\tilde{\lambda}_n^r$& refresh rate of item $n$ \\
  $B_{n}$ &  subset of $\tilde{\Omega}$ where no item in $\Nclosed{n}$ is cached\\ 
  $B_{n,m}$ &  subset of $\tilde{\Omega}$ where closest item to $n$ in cache is $m$\\ 
 $H$ & hit rate \\
 \hline
  \multicolumn{2}{c}{{\bf RND-TTL}} \\ \hline
  $\Omega$ & set of all potential sets of cached items in RND-TTL \\ 
  $\boldsymbol{\pi}$ & RND-TTL stationary distribution \\ 
  $ \lambda^i_n$ & insertion rate of item $n$ given that $n$ is not cached \\
  $\lambda^r_n$ & timer refresh rate of item $n$ given that $n$ is cached\\
  $X_{n}(t)$& $1$ if item $n$ is in cache at time $t$ and $0$ otherwise \\ 
  $o_n$& fraction of time item $n$ spent in the cache \\ 
  $h_n$& hit probability of item $n$ \\ 
  $T$ & initial timer duration \\ 
  \hline
\end{tabular}  
\end{table}

\section{Background, Notation and Assumptions}
\label{s:background-notation}


 \subsection{Similarity Caching} \label{sec:simcache}

 In \textbf{similarity search} systems, users can query a 
 remote server, storing a set of items $\mathcal{I}$, to send the $k$ most similar items to a given item $n$, according to a specific definition of similarity.
 In practice, items are often represented by vectors in $\mathbb{R}^{d}$ (called embeddings)~\cite{mcauley2015image} so that the dissimilarity cost, $\mathrm{dis(.,.)}: \mathcal{I}^2 \xrightarrow[]{} \mathbb{R}^{+}$, can be selected to be an opportune distance between the embeddings.
A cache, that stores a small fraction of the catalog~$\mathcal{I}$, could be deployed close to the users to reduce the fetching cost of similarity searches.
The seminal papers~\cite{falchi2008metric,pandey2009nearest} suggest the cache may answer a request using a local subset of items potentially different from the true closest neighbors to further reduce the fetching cost while still maintaining an acceptable dissimilarity cost. They refer to such caches as \textbf{similarity caches}.

\subsubsection{SIM-LRU}
 One of the popular dynamic similarity caching policies is SIM-LRU \cite{pandey2009nearest}. This policy maintains an ordered list of $C$  key-value pairs. A key is the embedding of some item~$n$ requested in the past and its corresponding value is a list containing the $k'\geq k$ closest items to $n$ in $I$. We denote by $S$ the set of keys stored in the cache.
Upon a similarity search for an item~$n$, SIM-LRU selects the closest local key to $n$, i.e., $\hat n \triangleq \argmin_{m\in S} \mathrm{dis}(n,m)$. If the dissimilarity cost between $n$ and $\hat n$ is smaller than or equal to a threshold $d>0$ ($\mathrm{dis}(n,\hat n)\leq d$), the request experiences an \textit{approximate hit}:\footnote{
    Note that we have an exact hit if $\hat n= n$.
} 
the cache replies to the request for $n$ selecting the $k$ closest items to $n$ among the $k'$ values stored for $\hat n$ and moves $\hat n$'s key-value pair to the front of the list. Otherwise, the request experiences a \emph{miss}: it is forwarded to the original server to retrieve the $k'$ closest items to $n$, out of which the closest $k$ are provided to the user. Upon a miss, the cache inserts the new key-value pair for $n$ at the front of the list and evicts the key-value pair at the bottom of the list. 
We observe how the use of key-value pairs in SIM-LRU essentially converts the search of $k$ closest items into the search of the closest key in the cache. For simplicity's sake, from now on we will just identify the items, their keys and the corresponding values and say for example that the cache replies to a request for $n$ with the closest item $\hat n$ in the cache.

\begin{algorithm}[tb]
\caption{RND-LRU \cite{pandey2009nearest} }\label{alg:RND-LRU}
\begin{algorithmic}[1]
\STATE \textbf{Input:}
\STATE A sequence of requests $\mathbf{R}_T = (r_1, \ldots, r_T)$
\STATE An initial cache state $\mathbf{S}_0 = (s_{0,1}, \ldots, s_{0,C})$
\STATE Threshold similarity $d$ and probabilities $(q_{n}(m))_{n,m \in \mathcal{I}^2}$
\STATE \textbf{Output:}
\STATE The cache state at each time step $t$.
\STATE \textbf{Algorithm:}
\FOR{$t = 1$ to $T$}
    \STATE Compute the closest item to $r_t$ in $S_{t-1}$ as $\hat{r}_t = \arg\min_{m \in S_{t-1}} \text{dis}(r_t, m)$. \label{line:closet-item-in-cache}
    \STATE Generate a uniform random number $u \in [0,1]$ 
    \IF{$u \leq q_{\hat{r}_t}(r_t)$  and $\text{dis}(\hat{r}_t, r_t) \leq d$} \label{line:condition-approximate hit}
        \STATE \textbf{Case 1:} Hit, encompassing exact hit and approximate hit
        \STATE $S_t \leftarrow$MoveToFront$(S_{t-1},\hat{r}_t)$ \label{line:action-approximate-hit}  
    \ELSE
        \STATE \textbf{Case 2:} Miss
        \STATE $S_t \leftarrow $InsertAtFront$(S_{t-1} \setminus s_{t-1,C}, r_t)$ \label{line:action-miss}  
    \ENDIF
\ENDFOR
\end{algorithmic}
\end{algorithm}

\subsubsection{RND-LRU}

RND-LRU\cite{pandey2009nearest} is a generalization of SIM-LRU. Similar to SIM-LRU, RND-LRU can use $\hat{n}$, the closest item in the cache to $n$, to reply to queries for $n$ with approximate matching. However, RND-LRU introduces an additional element of randomness to determine whether $\hat{n}$ should be effectively used to fulfill the considered query. This randomness is encompassed in RND-LRU parameters $(q_{m}(n))_{n,m\in \mathcal{I}^2}$. For every two items $n$ and $m$, $q_{m}(n)$ denotes the probability that a candidate item $m$ is effectively used to reply a query for $n$, given that $m$ is the closest item to $n$ in cache, $m=\hat{n}$, and $\mathrm{dis}(n ,m)\leq d$. Throughout this work, we let $q_n(n)=1$.

The details of RND-LRU are presented in Algorithm~\ref{alg:RND-LRU}.
 Upon receiving a request for item $n$, RND-LRU locates the closest item to $n$ in the cache, denoted as $\hat{n}$ (line \ref{line:closet-item-in-cache}).
The case $\hat n =n$ corresponds to an exact hit. Otherwise, if $\hat n \neq n$ and 
$\textrm{dis}(\hat{n},n) \leq d$, RND-LRU randomly yields an approximate hit, with a probability of $q_{\hat{n}}(n)$. Both cases are captured in lines \ref{line:closet-item-in-cache}-\ref{line:action-approximate-hit}. A random sample $u$ is generated uniformly at random in the interval between 0 and 1. If $u \leq q_{\hat{n}}(n)$ and $\textrm{dis}(\hat{n},n) \leq d$  we have a hit (line \ref{line:condition-approximate hit}). 
Note that hits include approximate hits and  exact hits, as the condition in line \ref{line:condition-approximate hit}   is true  when $\hat{n}=n$ ($u \leq q_n(n)=1$ and $\textrm{dis}(n,n)=0$).
After a hit, 
the cache is rearranged by moving $\hat{n}$ to the front of the list (line \ref{line:action-approximate-hit}). Alternatively, if $\textrm{dis}(\hat{n},n) > d$ or $u > q_{\hat{n}}(n)$  
we have a miss: 
RND-LRU evicts the least recently used item at the bottom of the list and inserts $n$ at the front (line \ref{line:action-miss}).

For practical purposes, the function $q_{m}(n)$ decreases with the  dissimilarity between $m$ and $n$, with $q_{n}(n)=1$ and $q_{m}(n)=0$ for $\mathrm{dis}(m,n)> d$. We retrieve the behavior of SIM-LRU when $q_{m}(n)=1$ if $\mathrm{dis}(m ,n)\leq d$ for any two items $n$ and $m$, and the behavior of LRU when $q_{m}(n)=1$ if $m=n$ and $q_{m}(n)=0$ otherwise. For acceptance probabilities between these extremes, RND-LRU yields a caching policy that lies between SIM-LRU and LRU.

\subsubsection{Hit Rate} In the realm of similarity caching, hits may refer either to exact hits towards the originally requested item or to approximate hits, encompassing requests served by items that are deemed similar to the originally requested item.    In this paper, we select the fraction of hits whether exact or approximate as a performance metric for similarity caches, and we refer to it as the hit rate.

\subsection{TTL Approximation for LRU Cache}
\label{ss:ttl-lru}

\noindent \textbf{TTL cache.} Time to Live (TTL) is a mechanism that restricts the duration of data within a network. TTL is used, for instance, by Content Delivery Networks (CDNs) and by the Domain Name System (DNS) to determine when cached items must be evicted~\cite{fofack2013modeling,alouf:hal-01258189}.
 In a TTL cache, each item stored in the cache is associated with a timer, initially set to a duration of $T$. When the timer for an item expires, that item is evicted from the cache. When a request is made for an item $n$ that is not currently in the cache, it results in a cache ``Miss'', and subsequently, $n$ is added to the cache. Conversely, if $n$ is already present in the cache and a new request is received for it, it is considered a cache ``hit'', and the timer associated with $n$ is reset to the original value of $T$. When the requests follow a Poisson process with request rate $\lambda_n$ for item $n$, and each request is independent from the previous ones, it follows that a hit occurs for an item if the inter arrival time between two requests for the same item is smaller than $T$. Thus, the hit probability for $n$, $h_n$, is given by: $h_n  = 1- \e^{-\lambda_n T}$. Since the flow of arrivals is Poisson,  the ``Poisson Arrivals See Time Averages'' (PASTA) property \cite{wolff1982poisson} implies that the probability~$\newoccupancy_n$ that an item $n$ is in the cache (i.e., its occupancy probability, or simply occupancy) is equal to the probability that a request for that same item experiences a hit, i.e. $\hit_n = \newoccupancy_n$.

\noindent \textbf{TTL approximation.} The TTL approximation is an efficient method to estimate, under IRM and a Poisson request process, the hit rate of the Least-Recently-Used (LRU) caching policy, a particular case of SIM-LRU where each item does not have other similar items but itself, ~\cite{irm-fagin-1977,che2002hierarchical}. This method has been supported by theoretical arguments in \cite{fricker2012versatile,jiang2018convergence,leonardi2015least}. The idea of the TTL approximation is to estimate the hit rate of an LRU cache with the hit rate of a TTL cache with parameter $T$ set to the specific value $t_C$ that guarantees that the expected number of cached items in the TTL cache is equal to $C$, namely: 
\begin{equation}\label{e:tC-Defintion}
    \sum_{n\in \mathcal{I}} \left(1-\e^{-\lambda_n  t_C} \right) = C . 
\end{equation} \noindent Parameter $t_C$ is referred to as the cache characteristic time.  The above expression allows us to compute $t_C$, e.g., by means of a bisection method. The hit probability for LRU is then approximated as: 
\begin{equation}
 \label{e:hit-rate-item-lru-che}
        h_n  \approx 1- \e^{-\lambda_n t_C}. 
\end{equation}


\noindent

\subsection{ Notation and  Assumptions}


We denote the set of items as $\mathcal{I}$. Items are often represented by vectors in $\mathbb{R}^{d}$ (called embeddings)~\cite{mcauley2015image} so that the dissimilarity cost, $\mathrm{dis(.,.)}: \mathcal{I}^2 \xrightarrow[]{} \mathbb{R}^{+}$, can be selected to be an opportune distance between the embeddings. We denote SIM-LRU and RND-LRU similarity threshold as $d$. In RND-LRU, we denote by $q_{m}(n)$ the probability that a candidate item $m$ is effectively used to reply a query for $n$, given that $m$ is the closest item to $n$ in cache and $\mathrm{dis}(n ,m)\leq d$. As in SIM-LRU and RND-LRU a request for item~$n$ could be served by any item within a distance $d$ of $n$, it is convenient to define the set of such items as $\Nclosed{n}\triangleq \{m\in \mathcal{I}: \; \mathrm{dis}(n,m) \leq d\}$. We call the elements in $\Nclosed{n}$ distinct from $n$ the \emph{neighbors} of $n$ and denote their set as $\Nopen{n}\triangleq \Nclosed{n} \setminus \{n\}$. For convenience, we also define in a similar way the sets $\Nopenin{m}{n}$ and $\Nclosedin{m}{n}$, respective subsets of $\Nopen{n}$ and $\Nclosed{n}$, designating items that are closer to $n$ than~$m$, i.e. $\Nopenin{m}{n}\triangleq \{ l\in \Nopen{n}: \; \mathrm{dis}(n,l)< \mathrm{dis}(n,m) \}$ and $\Nclosedin{m}{n}\triangleq \{ l\in \Nclosed{n}: \; \mathrm{dis}(n,l)< \mathrm{dis}(n,m) \} $.  As commonly used, $\mathds{1}(A)$ stands for the indicator function that $A$ is true. 

Throughout the paper, the following assumptions hold true.
\begin{assumption} 
\label{assum-IRM-request-process}
   Requests follow a Poisson process with request rate $\lambda_n$ for item $n$, and each request is independent from the previous ones, i.e., requests follow the Independent Reference Model (IRM)~\cite{irm-fagin-1977}.
\end{assumption}
\begin{assumption}
\label{assum-RND-TTL-StationaryDistribution}
    The state of RND-LRU is a continuous time Markov chain in $\tilde \Omega$ that has a limiting distribution.  
\end{assumption}

\begin{assumption} 
\label{assum:strict-order-neighbors}
    Items in $\Nopen{n}$ can be strictly ordered according to their dissimilarity with respect to $n$, i.e., for any $m,l\in \Nopen{n}$ and $m\neq l$, we have $\mathrm{dis}(n,m)\neq \mathrm{dis}(n,l)$.
\end{assumption}

If Assumption \ref{assum:strict-order-neighbors} does not hold, 
we can introduce an arbitrary order for items with the same dissimilarity.

\section{RND-TTL Approximation for Similarity Caching}
\label{s:TTL-SIM-LRU}


We introduce the RND-TTL approximation method, which aims to estimate the hit rate of RND-LRU by utilizing a RND-TTL cache that emulates the behavior of RND-LRU. Firstly, in Section \ref{ss:RND-TTL-sim-caching}, we provide a comprehensive description of the RND-TTL cache model, highlighting its specific characteristics. We also derive the occupancy probability, which represents the fraction of time spent by an item in the RND-TTL cache, and compute the probability of an item experiencing a hit under Assumption~\ref{assum-IRM-request-process}. Secondly, in Section~\ref{ss:RND-LRU-RND-TTL}, we demonstrate how the RND-TTL model captures the dynamics and behavior of RND-LRU. Thirdly, in Section~\ref{ss:RND-TTL-approximation}, we present a method for selecting appropriate parameters for the RND-TTL cache to estimate the hit rate of RND-LRU.

\subsection{The RND-TTL Caching Model}
\label{ss:RND-TTL-sim-caching}

\subsubsection{A TTL-based similarity caching policy} 
In a TTL cache, each item $n$ is associated with a timer initially set to $T_n$. Timers decrease over time, and an item is evicted from cache when its timer reaches zero.   
 Each time the item is requested, 
if an exact cache hit occurs, the item is served and its
TTL is reset.  Otherwise, we compute the closest item to $n$ in the cache, $\hat{n}$. If $\hat n $ is within a distance $d$ of $n$, i.e. $\hat n \in \Nopen{n}$, it can serve as a potential candidate to fulfill the request for $n$. In this case, a uniform random number $u \in [0,1]$ is generated. 
One of the following two scenarios occur:
\begin{itemize}
\item If $u < q_{\hat{n}}(n)$ and there is a candidate $\hat{n}$  to serve the request, such candidate is effectively used, corresponding to an approximate hit.  
\item If $u 
> q_{\hat{n}}(n)$ or there is no candidate to serve the request, we have a cache miss, and item $n$ is inserted into the cache. 
\end{itemize}
In all cases, the item used to serve the request has its 
TTL reset to $T_n$.



Note that RND-LRU and the above described TTL-based policy have in common the rules used to determine when and if items should be used to serve a given request.  However, whereas an item in an RND-LRU cache can be evicted as a result of the arrival of a request  that it cannot serve, 
in the TTL-based caches evictions occur after TTL reaches zero. In addition, a refresh in an  RND-LRU cache corresponds to a ``move to front'' operation, whereas in  a TTL-based cache it corresponds to a TTL reset.

Algorithm~\ref{alg:TTL-SIM2} in~\ref{app:algorithm-TTL} can be used to simulate the above policy.   Note that TTL-based caches are versatile, as we can adjust the hit probability of a file by controlling the time a file is kept in cache. In this paper, we leverage such versatility by considering a TTL-based model inspired by the above TTL-policy, that we refer to as RND-TTL. 

\subsubsection{RND-TTL and similarity caching} 
\label{sec:relaterndttltosimilarity}
The RND-TTL cache model is parameterized by the vector of TTLs 
 $\bm{T}=(T)_{n\in \mathcal{I}}$ and by two additional vectors $\boldsymbol{\lambda^r}=(\lambda_n^r)_{n\in \mathcal{I}}$, $\mathbf{p^i}=(p_n^{i})_{n\in \mathcal{I}}$ as described next.

The parameters $\bm{T}$ and $\boldsymbol{\lambda^r}$ dictate how long items remain in the cache. 
In the RND-TTL cache, each item is assigned a timer upon its insertion in the cache, and is evicted from the cache when its timer expires. Item $n$'s timer is initialized with the duration $T_n$ and is reset to $T_n$ over time according to a Poisson process with rate $\lambda_n^r$ (the superscript ``$r$'' refers to ``reset'' or ``refresh'').



Upon receiving a request for an item $n$, either $n$ is in the cache and is used to fulfill the request or it is not in the cache which gives rise to the two following possible scenarios:
    \begin{itemize}
     \item The request for $n$ results in a cache miss and consequently item $n$ is inserted into the cache. This scenario occurs with probability $p_n^i$ (the superscript ``$i$'' refers to ``insertion'').
     \item The request is fulfilled by the nearest item to $n$ in the cache, which occurs with probability $1-p_n^i$. 
 \end{itemize}

Note that the above model is inspired by the behavior of the TTL-based policy described in the previous section, while assuming that the dynamics of items are decoupled from each other as in traditional TTL systems. Indeed, $p_n^i$ and $\lambda_n^r$ can be set according to the modeling purposes. In particular, following the TTL approximation, in the rest of this paper we assume a single TTL value $T$ for all items, $T_n=T$, for all $n \in \mathcal{I}.$ Then,  we show that  $T$, $p_n^i$ and $\lambda_n^r$ can be set in a way to capture the behavior of RND-LRU, with the coupling between items reflected through a parametrization of those values.

\subsubsection{Computing occupancies} 
We are interested in computing the fraction of time $o_n$ spent by item $n$ in the RND-TTL cache in the stationary setting. Let $(X_n(t))_{t\geq 0}$ be the stochastic process taking value $1$ when item $n$ is in the cache and $0$ otherwise. The occupancy $o_n$ is formally written as follows.
\begin{align}
    o_n \triangleq \lim_{T\to +\infty} \frac{1}{T} \int_{0}^{T} \mathds{1}(X_n(u)=1) \;du. 
\end{align}

\begin{proposition}[Occupancy]
\label{prop:occupancy-RND-TTL}
    Under Assumption \ref{assum-IRM-request-process}, the occupancy in the RND-TTL cache of item $n$ is expressed as: 
\begin{align}
\label{e:occupancy-simlru-che}
    o_n =   \myon,
\end{align}
where 
\begin{align}\label{e:lambda-entry-def}
\lambda_n^i &= \lambda_n \cdot p_n^i.
\end{align}
\end{proposition}

\begin{proof}  The result follows from a renewal argument, where $\E{\Ton}$ and $\E{\Toff}$ are the mean time that an item resides on and off the cache, per cycle,
\begin{align}
    o_n = \frac{\E{\Ton}}{\E{\Toff} + \E{\Ton}} = \left(\E{\Toff} \cdot\frac{1}{\E{\Ton}}   +1  \right)^{-1}.
\end{align}
In the above expression, 
$\E{\Ton}$ is the mean duration of a busy period of an M/D/$\infty$ queue with   arrival rate and mean residence time  given by $\lambda_n^r$ and $t_C$,  respectively,
\begin{equation}
    \E{\Ton} = \frac{1}{\lambda_n^r} \left(e^{\lambda_n^r t_C} -1\right).
\end{equation}
$\E{\Toff}$ is the mean time to insert an item after it is removed,
\begin{equation}
    \E{\Toff} = \frac{1}{\lambda_n^i}.  
\end{equation}
For additional details  we refer the reader to~\ref{app:proofoccupancy}. 
\end{proof}

We stress that under Assumption~\ref{assum-IRM-request-process} the stationary distribution of $X_n$ is given by the occupancy, namely, $\lim_{t\to \infty} \Proba{X_{n}(t)=1}= o_n$ and $\lim_{t\to \infty} \Proba{X_{n}(t)=0}= 1-o_n$, thanks to \cite[Thm. 3.4.4]{ross1995stochastic}.

\subsubsection{Set of cached items distribution} 



Let $S$ denote the set of cached contents. In TTL-based policies such as RND-TTL, $S$ belongs to the power set of the catalog $\mathcal{I}$, i.e., $S \in \Omega$, where $\Omega = 2^\mathcal{I}$ and $2^\mathcal{I}$ denotes the power set of $\mathcal{I}$. Observing that the caching decisions in the RND-TTL cache are independent across items, it follows that for any set of cached items $S\in \Omega$, the corresponding stationary probability $\pi_S$ can be computed as follows:
\begin{align}\label{e:stationary-distribution-RND-TTL}
\pi_S= \prod_{n\in S} o_n \cdot \prod_{m\notin S} (1-o_m). 
\end{align}

\subsubsection{Item hit probability} 
We now give an explicit expression for the hit probability for each item in the RND-TTL cache. 

\begin{proposition}[Item hit probability]
\label{prop:hit-rate-RND-TTL}
Under Assumption \ref{assum-IRM-request-process}, the hit probability $\hit_{n}$ for item $n$ in the RND-TTL cache is given by: 
\begin{equation}\label{e:hit-rate-rnd-TTL}
   \hit_{n} =  o_n + (1-o_n)\cdot (1-p_n^i)\; , 
\end{equation}
where $o_n$ is given in \eqref{e:occupancy-simlru-che}. 
\end{proposition}

\begin{proof}
    See \ref{app:PASTA} and \ref{app:hit-rate-RND-TTL}. 
    
\end{proof}

The RND-TTL cache can be seen as a generalization of the TTL cache as the latter can be obtained when two conditions are met: $(i)$ 
$p_n^i=1$ for each item $n$, and $(ii)$ the timer refresh process of each item $n$ coincides with its request process. The hit rate of the TTL cache can be retrieved from \eqref{e:hit-rate-rnd-TTL} and \eqref{e:occupancy-simlru-che} by letting $p_n^i = 1$ and  $\lambda_n^r = \lambda_n^i = \lambda_n$. Equations \eqref{e:hit-rate-item-lru-che}, \eqref{e:occupancy-simlru-che} and \eqref{e:hit-rate-rnd-TTL} are then all equivalent.



While we have described the RND-TTL cache model using parameters $T$, $\boldsymbol{\lambda}^{r}$ and $\mathbf{p}^{i}$, in what follows, it will be more convenient to retain as parameters $T$, $\boldsymbol{\lambda^r}$, and $\boldsymbol{\lambda^i}=(\lambda_n^i)_{n\in \mathcal{I}}$ (see \eqref{e:lambda-entry-def}). 

\subsection{Relation between RND-LRU and RND-TTL}
\label{ss:RND-LRU-RND-TTL}

Using the RND-TTL cache to estimate the hit rate of RND-LRU is analogous to using the TTL cache to approximate the hit rate of LRU. 
Both RND-TTL and TTL enable the decoupling of caching decisions across items, with the goal of capturing the behavior of an item $n$ in terms of its insertion and eviction from the cache, independently of other items. We revisit the concepts of timer expiration, insertion policy, and timer re-initialization in the TTL cache and the RND-TTL cache and establish their connection to the caching decisions of LRU and RND-LRU, respectively.

{\noindent \textbf{Timer expiration.}} In both RND-LRU and LRU, an item is evicted from the cache when it is no longer among the $C$ recently used items. This behavior is captured and represented in RND-TTL and TTL caches by assigning a timer with a duration of $T$ to each cached item. An item is then evicted upon expiration of its timer. 


{\noindent \textbf{Insertion in the cache.}} In LRU/TTL cache, a non cached item $n$ is always inserted into the cache when it receives a request. It follows that the insertion rate for $n$, when it is not cached, is equal to its request rate $\lambda_n$ for both LRU and TTL. However, in RND-LRU, this is not the case as a non-cached item $n$ can be served by a similar item already in the cache. As a result, when $n$ is not in the cache, the \textbf{insertion rate} for item $n$ in RND-LRU is smaller or equal to $\lambda_n$. In RND-TTL, the parameter $\lambda_n^i$ serves as the insertion rate for item $n$ when it is not cached, allowing RND-TTL to capture the insertion behavior of item $n$ in the RND-LRU cache by tuning $\lambda_n^i$ accordingly.

{\noindent \textbf{Timer re-initialization.} }In LRU, when a cached item $n$ receives a request, it is refreshed by being moved to the front of the list. This behavior is captured in TTL by re-initializing the timer for item $n$. It follows that the \textbf{refresh rate} for $n$, when it is in the cache, is equal to $\lambda_n$ for both LRU and TTL. However, in the case of RND-LRU, the refresh process is not solely based on its own request. Item $n$ might also be refreshed when its neighboring items receive requests. As a result, in RND-LRU, when $n$ is cached, the \textbf{refresh rate} of $n$ is greater than or equal to $\lambda_n$. In RND-TTL, the parameter $\lambda_n^r$ determines the rate at which $n$'s timer is re-initialized when $n$ is cached. By appropriately adjusting $\lambda_n^r$, RND-TTL can capture the refresh operation of an item in RND-LRU.

In Section \ref{ss:RND-TTL-approximation}, we examine the insertion rate and refresh rate of an item in RND-LRU in detail, which allows us to derive guidelines for selecting the parameters $\boldsymbol{\lambda^i}, \boldsymbol{\lambda^r}$, and $T$ for the RND-TTL approximation.

\subsection{RND-TTL approximation to RND-LRU}
\label{ss:RND-TTL-approximation}
We propose an extension of the TTL approximation, namely the RND-TTL approximation, for estimating the hit rate of RND-LRU under Assumption \ref{assum-IRM-request-process}. We recall that the TTL approximation approximates the hit rate of LRU with the hit rate of a TTL cache and it has received theoretical support in \cite{fricker2012versatile,jiang2018convergence,leonardi2015least}. The RND-TTL approximation chooses the parameters of the RND-TTL cache model, namely $\boldsymbol{\lambda^i}$ and $\boldsymbol{\lambda^r}$, based on how the insertion rate and refresh rate for an item in RND-LRU depend on the limiting distribution of the Markov chain representing RND-LRU cache's state. Furthermore, $T$ is chosen in a similar manner to its selection in the TTL approximation, ensuring that the expected number of cached items is equal to~$C$. 


\begin{proposition}[RND-LRU insertion rate]
\label{prop:entry-rate-RND-LRU}
    Under Assumptions \ref{assum-IRM-request-process}, \ref{assum-RND-TTL-StationaryDistribution} and \ref{assum:strict-order-neighbors}, the insertion rate of item $n$ in RND-LRU, $\tilde{\lambda}_{n}^i$, is expressed as: 
    \begin{align}
        \tilde{\lambda}_{n}^i = \lambda_n \left(\Proba{B_{n}} + \sum_{m\in \Nopen{n}} (1-q_{m}(n)) \Proba{B_{n,m}}\right),
    \end{align}
    where 
    \begin{align}
    & B_{n}=\{ S\in \tilde \Omega: \; S\cap \Nclosed{n}=\emptyset \},\\
    & B_{n,m}=\{ S\in \tilde \Omega: \; m\in S , S \cap \Nclosedin{m}{n}=\emptyset\}\,, \, \forall m\in \Nopen{n}\,,
    \label{e:Bnm}
    \end{align}
and $\tilde \Omega$ refers to the cache state space with RND-LRU.
\end{proposition}

\begin{proof}
    See  \ref{app:PASTA} and \ref{app:proof-entry-rate-RND-LRU}.
    
\end{proof}

\begin{proposition}[RND-LRU refresh rate]
\label{prop:refresh-rate-RND-LRU}
    Under Assumptions \ref{assum-IRM-request-process}, \ref{assum-RND-TTL-StationaryDistribution} and \ref{assum:strict-order-neighbors}, the refresh rate of item $n$ in RND-LRU, $\tilde{\lambda}_n^r$, is expressed as: 
   \begin{align}
&\tilde{\lambda}_{n}^r = \sum_{m\in \Nclosed{n}} q_{n}(m) \lambda_m \,\Proba{B_{m,n}},
\end{align}
where $B_{m,n}$ is defined as in \eqref{e:Bnm} for every $m\in \Nclosed{n}$.
\end{proposition}
\begin{proof}
    See \ref{app:PASTA} and \ref{app:proof-refresh-rate-RND-LRU}. 
    
\end{proof}

{ \noindent \textbf{Choice of} $\boldsymbol{\lambda^r}$\textbf{,} $\boldsymbol{\lambda^i}$ \textbf{and} $\mathbf{T}$.} 
We have derived in Propositions~\ref{prop:entry-rate-RND-LRU} and~\ref{prop:refresh-rate-RND-LRU} the insertion rate $\tilde{\lambda}_n^i$ and the refresh rate $\tilde{\lambda}_n^i$ of item $n$ under the RND-LRU policy, respectively. The same rates for item $n$ in the RND-TTL model are $\lambda_n p_n^i(1-o_n)$ and $\lambda_n^r o_n$. However, the parameters $p_n^i$ and $\lambda_n^r$ are yet to be set. In order to approach the performance of RND-LRU, we propose using Propositions~\ref{prop:entry-rate-RND-LRU} and~\ref{prop:refresh-rate-RND-LRU} to set the values of the parameters $p_n^i$ and $\lambda_n^r$, provided the probabilities of the defined sets are computed according to the stationary distribution over the state space $\Omega$ of the RND-TTL cache.

For the insertion rate, we write
\begin{align}
\nonumber
\lambda_n p_n^i(1-o_n) &= \lambda_n \!\left(\Proba{B_{n}} + \!\sum_{m\in \Nopen{n}}\! (1-q_{m}(n)) \Proba{B_{n,m}}\right)
\end{align}
with $B_{n}$ and $B_{n,m}$ for $m\in \Nopen{n}$ redefined in the state space $\Omega$ of RND-TTL. To write the probabilities of the sets we rely on the independence of caching decisions in RND-TTL and its stationary distribution given in \eqref{e:stationary-distribution-RND-TTL}. It readily comes that the parameter $p_n^i$ is to be set as follows
\begin{align}
\label{e:pne-independence}
&   p_n^i =\!\!\prod_{m \in \Nopen{n} }\!\! (1-\newoccupancy_m) + \!\!\sum_{m\in \Nopen{n}}\!\! (1-q_m(n))\, \newoccupancy_{m} \!\!\prod_{j \in \Nopenin{m}{n}}\!\! \left(1-\newoccupancy_{j}\right).
\end{align}
For the refresh rate, we write
\begin{align}
\nonumber
\lambda_n^r o_n &=  \sum_{m\in \Nclosed{n}} q_{n}(m) \lambda_m \,\Proba{B_{m,n}}
\end{align}
with $B_{m,n}$ for $m\in \Nclosed{n}$ redefined in the state space $\Omega$. Leveraging the independence of caching decisions in RND-TTL and its stationary distribution given in \eqref{e:stationary-distribution-RND-TTL}, we set the parameter $\lambda_n^r$ as follows
\begin{align}
\label{e:lambda-refreshement-def}
&\lambda_n^r =  \lambda_n + \sum_{m\in \Nopen{n}} q_n(m) \lambda_m\prod_{j \in \Nclosedin{n}{m}} \left(1-\newoccupancy_{j}\right)  .
\end{align}


\noindent Given $\boldsymbol{\lambda^r}$ and $\boldsymbol{\lambda^i}$, $T$ is chosen from the set $T_C(\boldsymbol{\lambda^r},\boldsymbol{\lambda^i})$, meaning that $T$ verifies: 
\begin{align}\label{e:tC-sim}
    \sum_{n\in \mathcal{I}} \myonn{\lambda_{n}^{r}}{\lambda_{n}^i}{T} = C.
\end{align}
Choosing $T$ in this way guarantees that the expected number of items in the RND-TTL cache is equal to the cache capacity $C$ similarly to the TTL approximation for LRU as in \eqref{e:tC-Defintion}.

{ \noindent \textbf{RND-LRU's hit rate approximation. }If we can compute values for $\boldsymbol{\lambda^r}$, $\boldsymbol{\lambda^i}$ and $T$ verifying the equations above, then using Proposition \ref{prop:hit-rate-RND-TTL} and \eqref{e:pne-independence}, we can estimate the hit rate of RND-LRU as:
\begin{align} \label{e:Hit-Rate-RNDLRU-Approx}
        H \approx \sum_{n\in \mathcal{I}} \lambda_n \cdot \hit_{n},
\end{align}
where
\begin{equation}\label{e:hit-rate-rnd-lru-independence}
   \hit_{n} =  o_n + \sum_{m\in \Nopen{n}} q_{m}(n) \cdot o_m \prod_{j\in \Nclosedin{m}{n}} (1-o_j). 
\end{equation}

In the following section, we present an iterative algorithm that enables numerical determination of the parameters $\boldsymbol{\lambda^r}, \boldsymbol{\lambda^i}$, and $T$ based on the aforementioned description.


\section{Algorithm for Finding Approximate Hit Probabilities}
\label{s:Algorithm-fixed-point}

The RND-TTL approximation suggests to choose the parameters $\boldsymbol{\lambda^r}$, $\boldsymbol{\lambda^i}$ and $T$ for the RND-TTL cache such that the following set of equations is verified: 
\begin{align}
\label{e:lambda-entry-independence}
    &\boldsymbol{\lambda^i} = \fentry(\mathbf{o})  ,\\ 
    \label{e:lambda-refreshement-independence}
    &\boldsymbol{\lambda^r}  = \frefresh(\mathbf{o}) ,\\ \label{e:vec-occupancy-system}
    & \mathbf{\occupancy} = \foccupancy(\boldsymbol{\lambda^r},\boldsymbol{\lambda^i},T)= (\foccupancyC(\lambda_n^r,\lambda_n^i,T))_{n\in \mathcal{I}}  , \\ \label{e:capacity_constraint} 
    & T\in  T_C(\boldsymbol{\lambda^r},\boldsymbol{\lambda^i}) \iff \;  \sum_{n \in \mathcal{I}} g(\lambda^r_{n},\lambda^i_{n},T) =C,
    \end{align}
where $\fentry = (E_n(\mathbf{o}))_{n\in \mathcal{I}}$ with $E_n(\mathbf{o})$ obtained by substituting \eqref{e:pne-independence} in \eqref{e:lambda-entry-def}, $\frefresh = (R_n(\mathbf{o}))_{n\in\mathcal{I}}$ with $R_n(\mathbf{o})$ obtained from \eqref{e:lambda-refreshement-def} and $\mathbf{\occupancy}$ is obtained from \eqref{e:occupancy-simlru-che} and $g$ is defined as
\begin{align}\label{e:occupancy-g-function}
g(x_1,x_2,x_3) \triangleq \myonn{x_1}{x_2}{x_3}. 
\end{align}
Combining \eqref{e:lambda-entry-independence}-\eqref{e:capacity_constraint}, we obtain a system of $3N+1$ equations in $3N+1$ unknowns, from which we can obtain in particular the occupancies and the duration $T$. 
Finally, once the occupancies are known, we can compute the vector of hit probabilities, $\mathbf{h}=(h_n)_{n\in \mathcal{I}}$, according to \eqref{e:hit-rate-rnd-lru-independence} and estimate the hit rate of RND-LRU according to \eqref{e:Hit-Rate-RNDLRU-Approx}.

In Section \ref{ss:Fixed-point-equation}, we demonstrate that finding a solution for the system of equations \eqref{e:lambda-entry-independence}-\eqref{e:capacity_constraint} corresponds to finding a fixed point of a function $\mathbf{G}(\mathbf{o})$, i.e., a root for the equation   $\mathbf{o}=\mathbf{G}(\mathbf{o})$. Moreover, we show a sufficient condition for $\mathbf{G}$ to have at least one fixed point.
In Section \ref{ss:Algorithm-Fixed-Point}, we propose an iterative algorithm for finding a fixed point of $\mathbf{G}$, following the form $\mathbf{o}(j+1)= (1-\beta )\mathbf{G}(\mathbf{o}(j)) + \beta \mathbf{o}(j)$ for $\beta \in [0,1)$. Additionally, we provide a sufficient condition for the existence of a unique fixed point for $\mathbf{G}$ and the convergence of our iterative algorithm towards this unique fixed point, given a well-tuned parameter $\beta$. In Section \ref{ss:Choice-Beta}, we suggest a practical approach for tuning the parameter $\beta$ of our fixed point algorithm.

\subsection{Fixed Point Equations}
\label{ss:Fixed-point-equation}

We denote the set $T_C(\fentry(\mathbf{o}), \frefresh(\mathbf{o}))$ as $T_C(\mathbf{o})$. We denote the capped simplex as $\Delta_C$ such that: 
\begin{equation}
    \Delta_C \triangleq \{ \mathbf{o}\in \mathbb{R}^{N}:\; 0 \leq o_n\leq 1, \; \sum_{n\in \mathcal{I}} o_n=C \ \}. \\
\end{equation}

\begin{lemma}[$T_C(\mathbf{o})$ is a singleton]
\label{lem:tc-function}
If the cover condition, given by:
 \begin{equation}\label{e:condition-existence-fixed-point}
     \forall \mathcal{M}\subset \mathcal{I}:  \# \mathcal{M} \leq C ,\quad \# \bigcup_{n\in \mathcal{M}} \Nopen{n} <N-C,
 \end{equation}
is satisfied, then $T_C(\mathbf{o})$ has a unique element for every $\mathbf{o}\in \Delta_C$, meaning that: 
    \begin{align}
        &\forall \mathbf{o} \in \Delta_C, \exists ! T_0\in \mathbb{R}^{+}:  \quad F(\mathbf{o}, T_0) = 0, \\ \label{e:F-capacity-function}
        & F(\mathbf{o},T) \triangleq \sum_{n \in \mathcal{I}} g(R_n(\mathbf{o}),E_n(\mathbf{o}),T) - C,
    \end{align}
where the symbol $\exists !$ refers to unique existence.  
\end{lemma}    
\begin{proof}
 See  \ref{app: tC-Function}.
 
\end{proof}
\begin{remark} If the cover condition in $\eqref{e:condition-existence-fixed-point}$ is not satisfied, it indicates that we are in a situation of practical uninterest. In this case, storing any $C$ items in the cache would be sufficient to serve the requests of at least $N-C$ items.

\end{remark}

From now on, we assume that the cover condition in \eqref{e:condition-existence-fixed-point} is verified. Therefore, thanks to Lemma \ref{lem:tc-function}, $T_C(\mathbf{o})$ is a singleton and we can define a function $t_C$ from $\Delta_C$ to $\mathbf{R}^{+}$, where for each $\mathbf{o}$, it associates the unique element in $T_C(\mathbf{o})$, i.e., 
\begin{align}\label{e:def-tC}
        t_C(\mathbf{o}) = T \iff F(\mathbf{o},T) = 0.
\end{align}
We also introduce the function $\mathbf{G}$ from $\Delta_C$ to $\Delta_C$ defined as:       
\begin{align} \label{e:Fixed-Point-G}
    \forall \mathbf{o} \in \Delta_C, \; \mathbf{G}(\mathbf{o}) & \triangleq \foccupancy\left(\frefresh(\mathbf{o}), \fentry(\mathbf{o}),t_C(\mathbf{o})\right) \\
    &=(g(R_1(\mathbf{o}),E_1(\mathbf{o}),t_C(\mathbf{o})),\cdots,g(R_N(\mathbf{o}),E_N(\mathbf{o}),t_C(\mathbf{o}))).
\end{align}
It follows that finding a solution for the system of equations \eqref{e:lambda-entry-independence}-\eqref{e:capacity_constraint} boils down to finding a fixed point of $\mathbf{G}$ within $\Delta_C$. 

For any sets $A$ and $B$, we denote the set of functions from $A$ to $B$ that are continuously differentiable as $\mathcal{C}^{1}(A \to B)$.

\begin{lemma}[Differentiability of $\ftc$] 
\label{lem:tc-class-C1}
The function $\ftc$ is continuously differentiable within the set $\Delta_C$, i.e., $\ftc \in \mathcal{C}^{1}(\Delta_C \to \Delta_C)$.
The gradient of $\ftc$ can be expressed as:
\begin{equation} \label{e:gradient-tC}
    \forall j \in \mathcal{I}, \; \frac{\partial \ftc}{\partial o_j}(\mathbf{o}) = - \frac{\partial F}{\partial o_j}\left(\mathbf{o},t_C(\mathbf{o})\right) \cdot \left(\frac{\partial F}{\partial T}\left(\mathbf{o},t_C(\mathbf{o})\right) \right)^{-1}, 
\end{equation}
where $F$ is expressed in \eqref{e:F-capacity-function}. 
\end{lemma}

\begin{proof}
    See \ref{app:tC-class-C1}.  
\end{proof}

\begin{proposition}[Fixed Point Existence] 
\label{prop:G-C1-Exist-Fixed-Point}
The function $\mathbf{G}$ is continuously differentiable within $\Delta_C$ and it has at least one fixed point in $\Delta_C$.
\end{proposition}

\begin{proof}
    It is evident that the functions $\foccupancy(\cdot,\cdot,\cdot)$, $\frefresh(\cdot)$, and $\fentry(\cdot)$ are continuously differentiable within $\left(\mathbb{R}^{+}\right)^{N}\cdot \left(\mathbb{R}^{+}\right)^{N} \cdot \mathbb{R}^{+} $, $\Delta_C$ and $\Delta_C$, respectively. Furthermore, according to Lemma \ref{lem:tc-class-C1}, we have that $t_C\in \mathcal{C}^1(\Delta_C\to \Delta_C)$. As a result, we conclude that $\mathbf{G}\in \mathcal{C}^{1}(\Delta_C\to\Delta_C)$. 
    Noting that  $\Delta_C$ is a non empty compact convex set,  Brouwer's fixed point theorem \cite{park1999ninety} implies  the existence of a fixed point for $\mathbf{G}$.
    
\end{proof}

Proposition \ref{prop:G-C1-Exist-Fixed-Point} indicates that when the cover condition \eqref{e:condition-existence-fixed-point} is satisfied, it is possible to find parameters for the RND-TTL cache model that verify the system of equations \eqref{e:lambda-entry-independence}-\eqref{e:capacity_constraint}. Therefore, we can apply the RND-TTL approximation to estimate the hit rate of RND-LRU.  



\begin{algorithm}[tb]
\caption{Fixed point method}\label{alg:cap}
\begin{algorithmic}[1]
 \renewcommand{\algorithmicrequire}{\textbf{Input:}}
 \renewcommand{\algorithmicensure}{\textbf{Output:}}
 \REQUIRE $C$,  $\mathbf{\lambda},$ $\mathrm{dis}(.,.),$ $d,$  $(q_{n}(i))_{(n,i)\in I^{2}}$, $\beta$,   stopping condition 
 \ENSURE  Estimation of $\mathbf{o}, \mathbf{h} , t_C$ 
 \\ 
 \textit{Initialization}: 
\STATE Obtain $t_{C}(0)$ such that $  \sum_{n\in I}\left(1-\e^{- \lambda_{n}\cdot t_{C}(0)} \right) = C$ \label{line:ini0}
\STATE $\mathbf{o}(0) \gets 1-\e^{- \mathbf{\lambda}\cdot t_C(0)}$ 
\label{line:ini1}
\STATE $\mathbf{h}(0) \gets f^{h} (\mathbf{o}(0))$  
\STATE $j \gets 1$
 
\WHILE{\textit{Stopping condition not satisfied}} 
 \label{line:work0}
    \STATE $\boldsymbol{\lambda^i}(j) \gets \fentry(\mathbf{o}(j-1))$ (See \eqref{e:lambda-entry-independence})  \label{line:lambda-e}
    \STATE $\boldsymbol{\lambda^r}(j) \gets \frefresh(\mathbf{o}(j-1))$  (See \eqref{e:lambda-refreshement-independence}) \label{line:lambda-r}
\label{line:work1}
    \STATE $\textrm{Obtain } \label{line:obtaintcj1}  t_{C}(j) \textrm{ such that}:  \sum_{n\in I} (\foccupancy( \boldsymbol{\lambda^i}(j) , \boldsymbol{\lambda^r}(j) , t_C(j)))_{n} =C$ (See \eqref{e:capacity_constraint},\eqref{e:vec-occupancy-system})
\label{line:obtaintcj} 
    \STATE $\mathbf{o}(j) \gets (1-\beta) \cdot \foccupancy( \boldsymbol{\lambda^i}(j) , \boldsymbol{\lambda^r}(j) , t_C(j))+ \beta \cdot \mathbf{o}(j-1)$  \label{line:occ1}  
    \STATE  $\mathbf{h}(j) = f^h(\mathbf{o}(j))$ (See \eqref{e:hit-rate-rnd-lru-independence})
    \STATE $j\gets j+1$
\ENDWHILE 
 
\RETURN $\mathbf{\hit}(j), \; \mathbf{\newoccupancy}(j), \; t_C(j)$  
\end{algorithmic} 
\end{algorithm}

\begin{figure}[t]
\centering
\begin{tikzpicture}[>=stealth]
    \node[rectangle, draw, rounded corners, minimum width=3cm, minimum height=1.5cm, align=center] (block1) {Compute occupancies  given $T$ and \\ insertion/refresh rates};
    
    \node[rectangle, draw, rounded corners, minimum width=3cm, minimum height=1.5cm, align=center, right of=block1, node distance=7cm] (block2) {Compute  insertion/refresh rates \\ given occupancies};

     \node[rectangle, draw, rounded corners, minimum width=5.5cm, minimum height=1.2cm, align=center, below of=block1, node distance=2.5cm] (block3) {Compute $T$ given \\ insertion/refresh rates};

  \draw[->, rounded corners=5mm, minimum width=5cm, bend left=30] (block2.south) to node[above] {$\bm{\lambda^i}$, $\bm{\lambda^r}$} (block3.east);

    \draw[->, rounded corners=5mm, minimum width=5cm, bend left=30, node distance=2cm] (block3.north) to node[above] {$\quad \quad T$} (block1.south);

    \draw[->, rounded corners=5mm, minimum width=5cm, bend left=40] (block1.east) to node[above] {$\bm{o}$} (block2.west);
    \draw[->, rounded corners=5mm, bend left=40] (block2.west) to node[below] {$\bm{\lambda^i}$, $\bm{\lambda^r}$} (block1.east);
\end{tikzpicture}
\caption{Essence of the fixed point algorithm.}
\label{f:fixedpoint}
\end{figure}
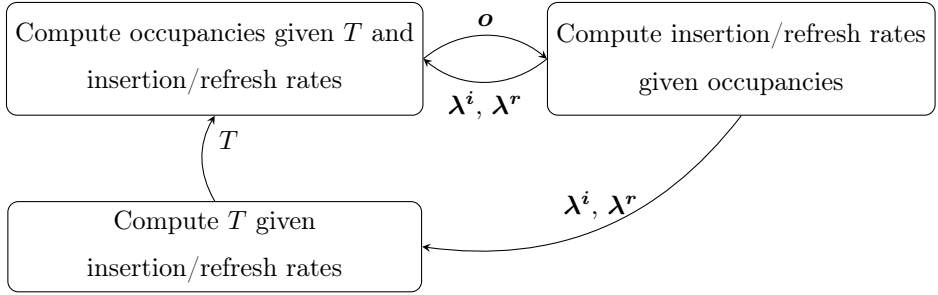

\subsection{ Fixed Point Algorithm}
\label{ss:Algorithm-Fixed-Point}

A natural approach to finding a fixed point of $\mathbf{G}$ is through an iterative method. This is illustrated in Figure~\ref{f:fixedpoint}. Starting with an initial guess $\mathbf{o}(0)$, we perform iterations of the form $\mathbf{o}(j+1) = \beta \mathbf{o}(j) + (1-\beta)\mathbf{G}(\mathbf{o}(j))$, where $\beta\in [0,1)$ \cite{mann1953mean}. A detailed version of these iterations is presented in Algorithm~\ref{alg:cap}.
Initially, we guess the occupancies $\mathbf{o}$, using LRU occupancies as a starting point. Specifically, we set $\mathbf{o}(0) = 1 - \e^{-\mathbf{\lambda} t_C(0)}$, where $t_C(0)$ satisfies \eqref{e:tC-Defintion} and $\sum_{n\in \mathcal{I}} o_n(0) = C$ (lines~\ref{line:ini0}--\ref{line:ini1}).
Then, we compute $\boldsymbol{\lambda^i}(1)$ and $\boldsymbol{\lambda^r}(1)$ using \eqref{e:lambda-entry-independence} and \eqref{e:lambda-refreshement-independence}, respectively (lines \ref{line:work0}--\ref{line:work1}).
The value of $t_C(0)$ is determined by solving the equation \eqref{e:capacity_constraint} using methods such as bisection or Newton's method.
Next, we calculate the new estimate of the occupancies $\mathbf{o}(1)$ (line \ref{line:occ1}).

The same procedure is repeated for subsequent iterations until a stopping condition is met. This condition could be, for example, the difference between the occupancies computed at consecutive iterations becoming smaller than a given threshold, or reaching the maximum number of iterations ($j\leq n_{\mathrm{iterations}}$).

Note that in Figure~\ref{f:fixedpoint} the boxes on the left-hand side, together with their inputs (insertion and refresh rates), represent the conventional perspective on caching. This involves using fixed rates, and computing item occupancies to estimate hit probabilities.  Under a TTL-based model, it also involves computing the characteristic time $T$ to approximate LRU, so that the sum of expected occupancies equals the cache capacity $C$. In contrast, the box on the right-hand side takes into consideration the unique nature of similarity caches. In similarity caches, the insertion and refresh rates are influenced by the current state of the cache, and these rates are determined as a function of the occupancies.


For a given value of $\beta$, the iterations of Algorithm \ref{alg:cap} are of the form: $\mathbf{o}(j+1) = \mathbf{G_{\beta}} (\mathbf{o}(j))$ such that: 
\begin{equation} \label{e:Defintion-G-Beta}
    \mathbf{G}_{\beta} (\mathbf{o}) \triangleq (1-\beta) \mathbf{G}(\mathbf{o}) + \beta \mathbf{o}.  
\end{equation}
The function $\mathbf{G_\beta}$ is continuously differentiable thanks to Proposition \ref{prop:G-C1-Exist-Fixed-Point}. We denote its Jacobian matrix as $\mathcal{J}_{\mathbf{G_\beta}}$. 

\begin{proposition}[Fixed Point Uniqueness and Convergence]
\label{prop:jacobian-convergence}
    If
    \begin{equation}\label{e:condition-unique-fixed-point}
       \exists \beta \in [0,1):\;  \text{sup}_{\mathbf{o}\in \Delta_C} \norm{\mathcal{J}_{\mathbf{G_\beta}}(\mathbf{o})} < 1, 
    \end{equation}
where $\norm{\cdot}$ can be any operator norm, then $\mathbf{G}$ has unique fixed point in $\Delta_C$ and Algorithm \ref{alg:cap} with parameter $\beta$ converges to this unique fixed point. 
\end{proposition}

\begin{proof}
Let $L= \text{sup}_{\mathbf{o}\in \Delta_C} \norm{\mathcal{J}_{\mathbf{G_\beta}}(\mathbf{o})}$, under \eqref{e:condition-existence-fixed-point} $\mathbf{G_\beta}\in \mathcal{C}^{1}(\Delta_C\to \Delta_C)$,  we deduce that $\mathbf{G_\beta}$ is Lipshitz with constant $L$ \cite{weaver2018lipschitz}, i.e., $\norm{\mathbf{G_\beta}(\mathbf{o_1})-\mathbf{G_\beta}(\mathbf{o_2})} \leq L \norm{\mathbf{o_1}-\mathbf{o_2}}$ for any $\mathbf{o_1},~ \mathbf{o_2}$. Taking advantage of the popular Banach fixed point Theorem \cite{meise1997introduction}, we deduce that $\mathbf{G_\beta}$ has a unique fixed point and Algorithm \ref{alg:cap} with parameter $\beta$ converges to this unique fixed point. Notice that for any $\beta$, the set of fixed points of $\mathbf{G}$ and $\mathbf{G_\beta}$ coincides, which concludes the proof.

\end{proof}


In practice, one can compute the norm of the matrix $\JacobianGB$ for few vectors $\mathbf{o}\in \Delta_C$ to get an idea of whether the sufficient condition in \eqref{e:condition-unique-fixed-point} is satisfied and then Algorithm \ref{alg:cap} with parameter $\beta$ converges to a unique fixed point. In the next proposition, we give an explicit formula for $\JacobianGB$ to ease its computation.

We denote by $\text{Diag}(x)$   an $N$-dimensional diagonal matrix, where the entries of the vector $x$ are positioned along its diagonal, and by  $I_N$ the $N$ dimensional identity matrix.  

\begin{proposition}[Computation of  $\JacobianGB$]\label{e:prop-Jacobian-GB}
    The Jacobian matrix $\JacobianGB$ has the following expression: 
\begin{equation} \label{e:Jacobian-G-Beta}
   \forall \mathbf{o}_0 \in \Delta_C, \quad    \JacobianGB(\mathbf{o}_0) = (1-\beta) \cdot \mathcal{J}_{\mathbf{G}}(\mathbf{o}_0) + \beta \cdot I_{N},
\end{equation}
where
\begin{align} 
 \nonumber \label{e:jacobian-G-Experession}
     \mathcal{J}_{\mathbf{G}}(\mathbf{o}_0) =  \text{Diag}\left(\partial_1 \foccupancy \right)&\cdot \mathcal{J}_{\frefresh}(\mathbf{o}_0)  + \text{Diag}\left(\partial_2 \foccupancy \right) \cdot \mathcal{J}_{\fentry}(\mathbf{o}_0)\\
     &- \frac{1}{\sum_{n\in \mathcal{I}} \partial_{3} \foccupancy_n} \partial_{3} \foccupancy^\intercal \cdot \left( \partial_1 \foccupancy \cdot \mathcal{J}_{\frefresh}(\mathbf{o}_0) + \partial_2 \foccupancy\cdot \mathcal{J}_{\fentry}(\mathbf{o}_0)  \right),
\end{align}
and $\mathcal{J}_{\frefresh}$ and $\mathcal{J}_{\fentry}$ are the Jacobian matrices of the functions $\frefresh$ and $\fentry$, respectively.
For $j\in \{1,2,3\}$, the vector $\partial_j \foccupancy$ has $n$ columns, where the $n$-th component denoted as $\partial_j g_n$ is given by $\frac{\partial g}{\partial x_j}(R_{n}(\mathbf{o}_0), E_{n}(\mathbf{o}_0), t_C(\mathbf{o}_0))$, with $g$ defined in~\eqref{e:occupancy-g-function}.
\end{proposition}

\begin{proof}
See \ref{app:proof-jac-GB}. 
\end{proof}

\subsection{Choice of $\beta$}
\label{ss:Choice-Beta} 

Our approach for selecting the value of $\beta$ for Algorithm \ref{alg:cap} is based on Proposition \ref{prop:jacobian-convergence}. Let $Y(\mathbf{o})$ be the set of values of $\beta$ in $[0,1)$ for which the spectral norm of $\JacobianGB$ is smaller than $1$, i.e., 

\begin{equation}
   Y(\mathbf{o}) \triangleq \{\beta\in [0,1):~  \norm{\JacobianGB(\mathbf{o})}_2 < 1 \}.
\end{equation}

\noindent Equation \eqref{e:condition-unique-fixed-point} in Proposition \ref{prop:jacobian-convergence} is equivalent to the set $\bigcap_{\mathbf{o}\in \Delta_C} Y(\mathbf{o})$ being non-empty. In other words, choosing the parameter $\beta$ of Algorithm \ref{alg:cap} from the set $\bigcap_{\mathbf{o}\in \Delta_C} Y(\mathbf{o})$ guarantees the convergence of Algorithm \ref{alg:cap} to the unique fixed point of $\mathbf{G}$.

We stress that the characterization of the sets $Y(\mathbf{o})$ and $\bigcap_{\mathbf{o}\in \Delta_C} Y(\mathbf{o})$ is difficult. For this reason, we proceed with a randomized approach. First, we randomly sample $d$ vectors from $\Delta_C$, $(\mathbf{o}_j)_{1 \leq j \leq d}$. Then, for each sampled vector $\mathbf{o}_j$, we compute a subset of values of $\beta$ leading to $\norm{\JacobianGB({\mathbf{o}_j)}}_2 < 1$. We denote the considered subset of $Y{(\mathbf{o}_j)}$ as $ \tilde{Y}(\mathbf{o}_j)$, where~$\tilde{Y}(\mathbf{o}_j) \subset Y(\mathbf{o}_j)$. Finally, we take the intersection $\bigcap_{j=1}^{d} \tilde{Y}(\mathbf{o}_j)$ as a set of candidate values for $\beta$. 
When we use a larger number of sampled vectors, $d$, the likelihood that the values of $\beta \in \bigcap_{j=1}^{d} \tilde{Y}(\mathbf{o}_j)$ satisfy the condition in~\eqref{e:condition-unique-fixed-point} increases. However, this also comes with the drawback of higher computational costs.


In the next proposition, we compute the aforementioned subset of $Y(\mathbf{o})$, $\tilde{Y}(\mathbf{o})$, based on the input vector $\mathbf{o}$. To this aim, we leverage the spectral radius of the Jacobian. Recall that the spectral radius of a matrix is defined as the maximum absolute value of its eigenvalues. We denote the spectral radius of matrix $M$ by $\rho(M)$, and its spectral norm by $\norm{M}_2=\sqrt{\rho{(M M^{\intercal})}}$.

\begin{proposition}[Properties of $Y(\mathbf{o})$]
\label{lem:effect-beta-spectral-norm}
Let $\gamma$ be the spectral norm of the Jacobian matrix   $\mathbf{G}(\mathbf{o})$ and let $\eta$ be the spectral radius of the matrix $\mathcal{J}_{\mathbf{G}}(\mathbf{o})+ \mathcal{J}_{\mathbf{G}}(\mathbf{o})^\intercal$,
\begin{align}
\gamma&=\norm{\JacobianG(\mathbf{o})}_2 = \sqrt{\rho(\JacobianG(\mathbf{o}) \JacobianG(\mathbf{o})^{\intercal} )}\\
\eta&=\rho(\mathcal{J}_{\mathbf{G}}(\mathbf{o})+ \mathcal{J}_{\mathbf{G}}(\mathbf{o})^\intercal).
\end{align}
If 
\begin{equation}\label{e:condition-subsest-Y}
            \eta^{2} - 4 \eta \gamma + 4 \geq 0 ,
\end{equation}
 then $Y(\mathbf{o})$ satisfies 
     \begin{align}
        &Y(\mathbf{o}) \supset \tilde{Y}(\mathbf{o})  \label{e:lem-Yo-a}
        \end{align}
        where
        \begin{equation}
             \tilde{Y}(\mathbf{o}) = (a,b)
        \end{equation}
        and
        \begin{align} 
        &a = \max\left(0, \frac{ 2 \gamma - \eta - \sqrt{\eta^{2} - 4 \eta \gamma + 4}}{2 \left(\gamma + 1\right)}\right), \\ \label{e:lem-Yo-b}
        &b = \min \left(1, \frac{ 2 \gamma - \eta  + \sqrt{\eta^{2} - 4 \eta \gamma + 4}}{2 \left(\gamma + 1\right)} \right).   
    \end{align}

\end{proposition}

\begin{proof}
    See \ref{app:properties}.  
\end{proof}

\begin{remark}
If, for a given $\mathbf{o}$, $\mathcal{J}_{\mathbf{G}}(\mathbf{o})$ is antisymmetric, i.e, $\mathcal{J}_{\mathbf{G}}(\mathbf{o})=-\mathcal{J}_{\mathbf{G}}(\mathbf{o})^{\intercal}$, then $\eta=0$ and \eqref{e:condition-subsest-Y} is verified. It follows from Proposition~\ref{lem:effect-beta-spectral-norm} that in this case $\tilde Y(\mathbf{o}) = \left(\frac{\gamma -1}{\gamma +1}, 1\right)$.
\end{remark}

For each vector $\mathbf{o}_j$, we compute the associated constants $\gamma_j = \gamma(\mathbf{o}_j)$ and $\eta_j = \eta(\mathbf{o}_j)$. Next we verify if \eqref{e:condition-subsest-Y} is satisfied for each pair $(\gamma_j, \eta_j)$. If this condition holds for all pairs, then using Proposition~\ref{lem:effect-beta-spectral-norm} we can determine a subset of  $\bigcap_{j=1}^{d} Y(\mathbf{o}_j)$, namely  $\bigcap_{j=1}^{d} \tilde{Y}(\mathbf{o}_j)$, and select a value for $\beta$ from within that subset to guarantee the convergence of the proposed fixed point algorithm.

\section{Numerical Evaluation}
\label{s:experiments}

\begin{table}[t]
   \caption{Parameters of the experiments.}
   \label{tab:parameters-experiments}
   \centering
\begin{tabular}{ |l|l|l| } 
 \hline
 Variable & Synthetic traces  & Amazon trace \\
 \hline
 \hline
 $\mathcal{I}$ & $[0..99]^{2}$ & Products  \\
 $N=|\mathcal{I}|$ & $10^{4}$ & $\approx 10^{4}$  \\
 $\lambda_n$ & \eqref{e:arrival-rates-square-two-hot-regions} & Empirical \\
 $\mathrm{dis}(\cdot, \cdot)$ &Euclidean distance & Euclidean distance \\
 $d$&  $1$ and $2$ & $300$ \\
 Number of requests $r$ & $2\cdot 10^{5}$  & $\approx 10^{5}$ \\ 
 Number of iterations Alg. \ref{alg:cap} & $25$ and $15$  & $40$\\ 
$q_n(m)$  & $\left(\mathrm{dis}(n,m)\right)^{-2}$ &   $\left(\mathrm{dis}(n,m)\right)^{-0.2}$ \\
 \hline 
\end{tabular}  
\end{table}

\begin{figure*}
    \begin{subfigure}{\linewidth}
  \centering
  \includegraphics[width=0.6\linewidth]{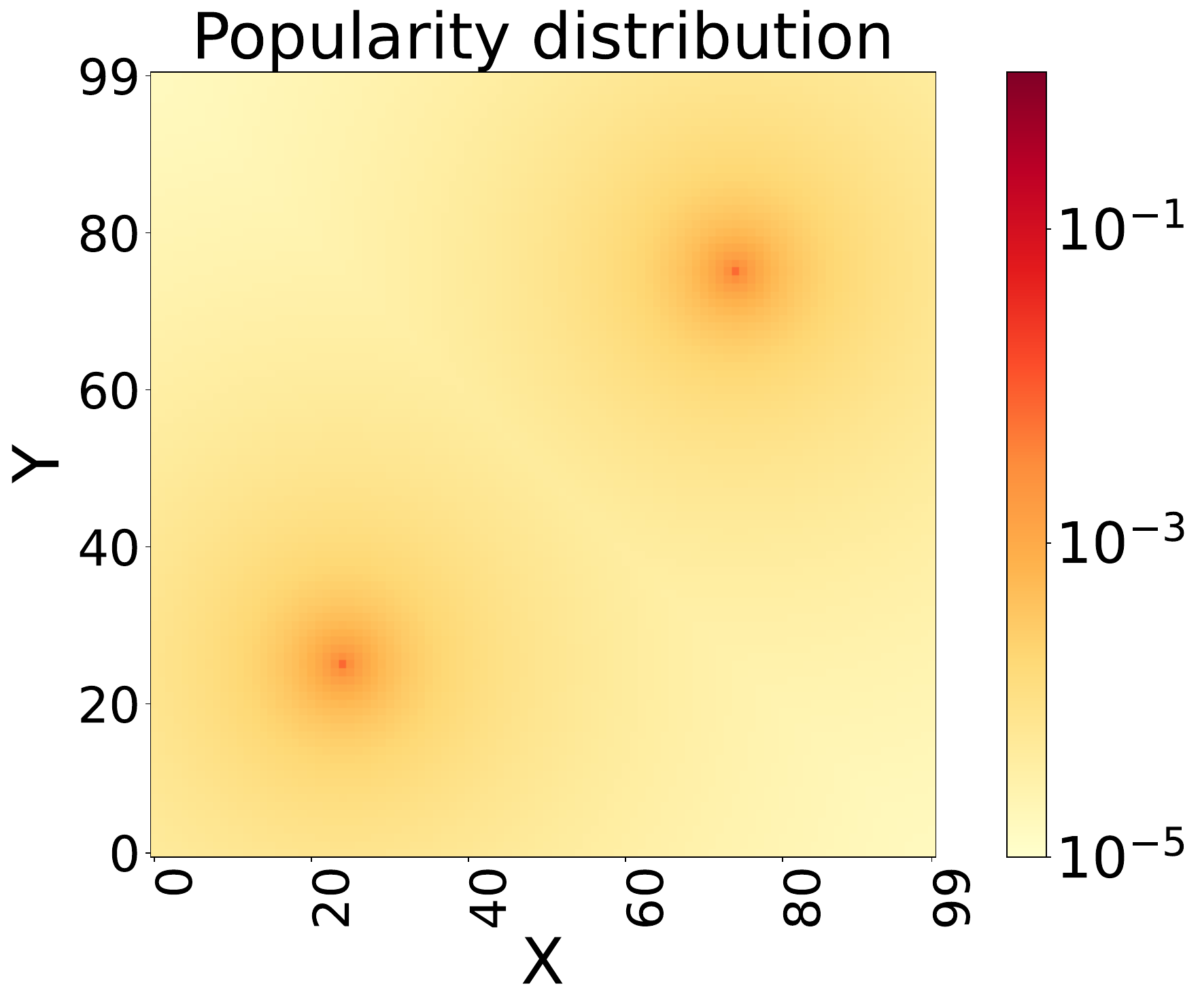}
\caption{$\alpha=1.4$}
\label{fig:popularity-alpha1.4}
\end{subfigure}
\hfill
\begin{subfigure}{\linewidth}
  \centering
  \includegraphics[width=0.6\linewidth,keepaspectratio]{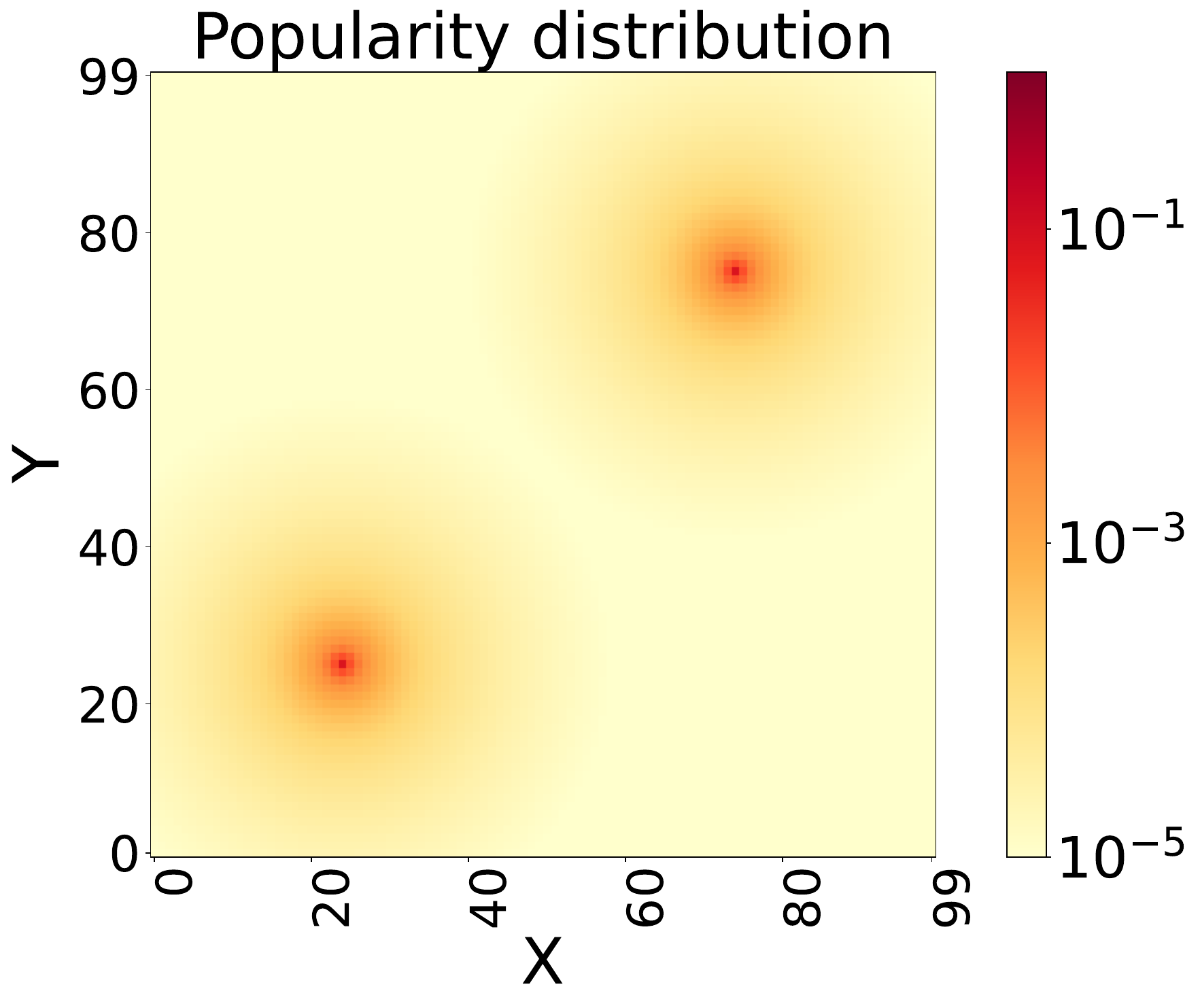}
\caption{$\alpha=2.5$}
\label{fig:popularity-alpha2.5}
\end{subfigure}
\begin{subfigure}{\linewidth}
  \centering
  \includegraphics[width=0.6\linewidth,keepaspectratio]{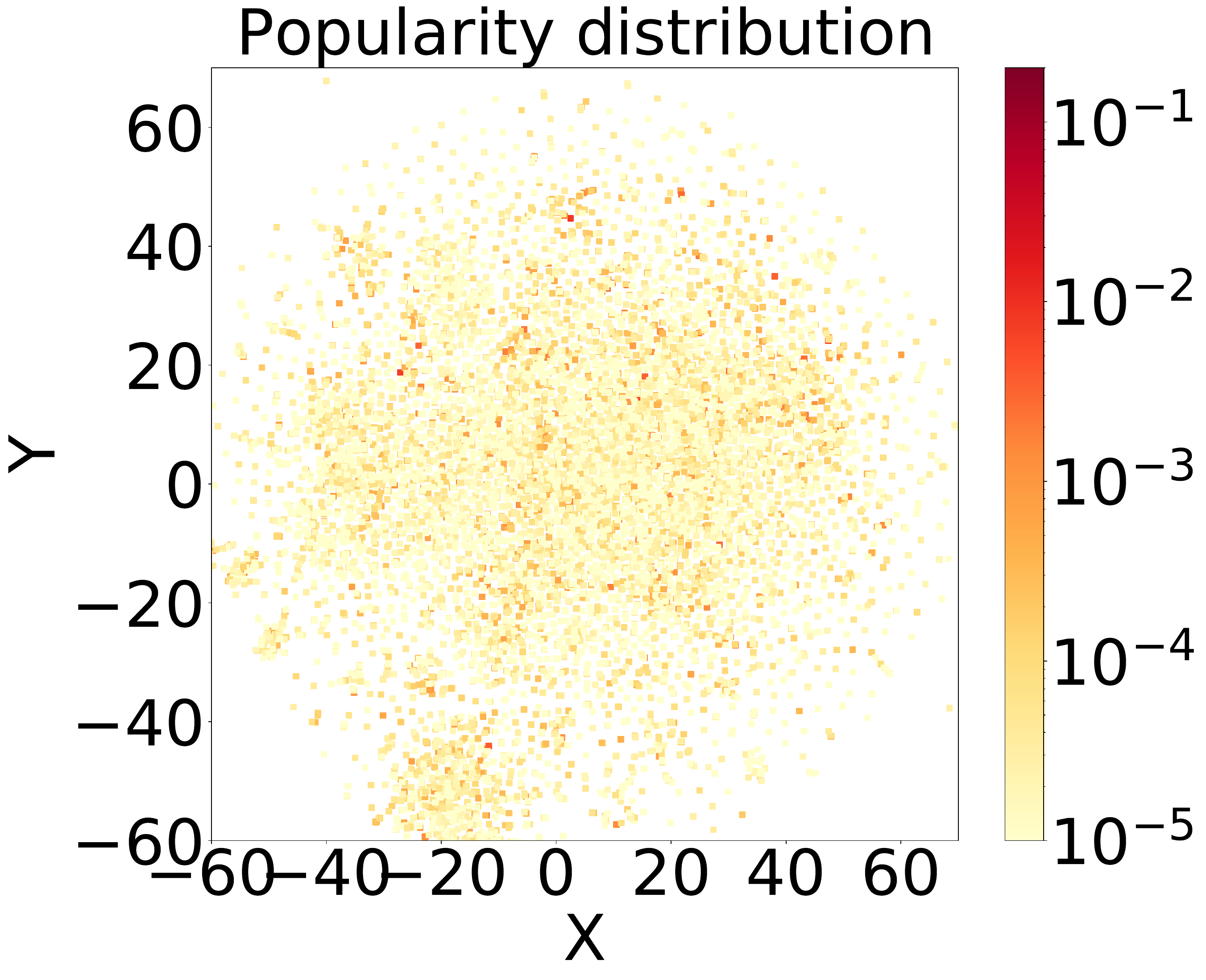}
\caption{Amazon trace}
\label{fig:popularity-amazon}
\end{subfigure}
    \caption{Spatial popularity distribution.}
    \label{fig:popularities}
\end{figure*}

 We assess the accuracy of our proposed RND-TTL approximation method by conducting experiments on both synthetic and real-world traces. The traces are described in Section \ref{ss:Experimental-Setting}. To evaluate our method, we compare our approach for estimating the hit rate of RND-LRU with other alternative solutions discussed in Section~\ref{ss:alternative solutions}. Subsequently, we analyze the approximation accuracy in Section~\ref{ss:RND-TTL approximation-experiments}. Our RND-TTL approximation technique involves solving a set of equations using an iterative fixed-point method outlined in Algorithm~\ref{alg:cap}. In Section \ref{ss:Convergence-Algorithm1-experiments}, we evaluate the convergence of Algorithm~\ref{alg:cap}.

\subsection{Experimental Setting} 
\label{ss:Experimental-Setting}

We evaluate the efficiency of the proposed fixed point method (Algorithm~\ref{alg:cap}) to predict the hit rate on synthetic traces and on an Amazon trace \cite{sabnis2021grades}.

{ \noindent \textbf{Synthetic traces.}} For the synthetic traces, each item corresponds to two features, characterized by a point in a grid, $\mathcal{I}=[0..99]^{2}$ (e.g. Fig.~\ref{fig:popularities}). The total number of items is $|\mathcal{I}|= 10^{4}$, and the dissimilarity function between items  $\mathrm{dis}(\cdot,\cdot)$ is the Euclidean distance. Neighbors of item $(x,y)$ at the same distance are ordered counterclockwise starting from the item to the right, i.e., from $(x+a,y)$ with $a>0$. The synthetic traces are generated in an IRM fashion~\cite{irm-fagin-1977}, where the popularity distribution for an item $n=(x,y)$ is given by 
\begin{equation}
\label{e:arrival-rates-square-two-hot-regions}
     p_{(x,y)} \sim  \left(\min \left \{\mathrm{dis}(n,(24,24)), \mathrm{dis}(n,(74,74)) \right\} +1\right)^{-\alpha} ,
\end{equation}
where $\alpha$ is a parameter controlling the skewness of the popularity distribution.
We generate $50$ synthetic streams for $\alpha=1.4$ and $\alpha=2.5$ and each stream is of $r=2\cdot 10^{5}$ requests for items in $\mathcal{I}$. Figures~\ref{fig:popularity-alpha1.4}, \ref{fig:popularity-alpha2.5} illustrate the popularity distribution in \eqref{e:arrival-rates-square-two-hot-regions} for $\alpha \in \{1.4, 2.5 \}$.

{ \noindent \textbf{Real world trace. }}For the Amazon trace, each item corresponds to an Amazon product. The request trace in ~\cite{sabnis2021grades} is generated by mapping every Amazon review for the item to an item request. Each item has been mapped to an Euclidean space of dimension $100$ using the technique in ~\cite{mcauley2015image}, where the Euclidean distance reflects dissimilarity between two items. Inspired by this methodology, we generate a corresponding IRM stream of requests matching the item popularity in the trace in \cite{sabnis2021grades}. Figure \ref{fig:popularity-amazon} shows a representation of item popularities in a 2D space through the t-SNE algorithm ~\cite{van2008visualizing}.

\subsection{Benchmarks and Alternative Approaches}
\label{ss:alternative solutions}

In what follows, we compare hit rate estimates provided by RND-LRU using  Algorithm~\ref{alg:cap} with the hit rate estimations for LRU and for the optimal static allocation. We also propose an alternative approach to estimate RND-LRU's hit rate.

\textbf{LRU.} The hit rate and the occupancy for an item $n$ are computed using \eqref{e:hit-rate-item-lru-che} and $t_C$ is deduced using the cache capacity constraint given by \eqref{e:tC-Defintion}.

\textbf{Optimal Static Allocation.} 
The maximum hit rate obtainable by a static allocation under similarity caching can be obtained solving a maximum weighted coverage problem. We consider, as in SIM-LRU, that each item can be used to satisfy any request for items closer than $d$.
The maximum weighted coverage problem takes as input a capacity $C$, a set of items ~$\mathcal{I}$, with $N=|\mathcal{I}|$, their corresponding weights $W=(w_n)_{n\in I}$ and a set of sets $R=\{R_1,\ldots, R_N \}$ such that $R_n \subset \mathcal{I}$.  The objective is to find a set $\sigma^{*}\subset \{ 1,\ldots , N\} $  such that: $ \sigma^{*} = \argmax_{\sigma \subset \{ 1,\ldots , N\}: |\sigma|\leq C  } \sum_{n\in \cup_{j\in \sigma } R_j}  w_n$.

Finding the best static allocation is equivalent to solving a maximum weighted coverage problem, with   
weights $w_n= \lambda_n$ for $n\in \mathcal{I}$, $C$ the cache capacity, and $R$ the set of neighbors for each item, i.e., $R= \{\Nclosed{n} \}_{n\in \mathcal{I}}$. 

The maximum weighted coverage problem is known to be NP-hard. In practice, a popular greedy algorithm guarantees a $(1- 1/e)$ approximation ratio~\cite{nemhauser1978analysis}.

The greedy algorithm operates as follows: initially, it selects the set $R_{c_1}=R_{c_1}^{0}$ with the largest coverage, where $c_1$ is determined by $c_1=\argmax_{n\in \mathcal{I}} \sum_{m \in R_n} \lambda_m$. Subsequently, the algorithm considers sets $(R_{n}^{1})_{n\in \mathcal{I}}$ defined as $R_{n}^{1} = R_{n}^{0} \setminus R_{c_1}^{0}$ in the next step, and it chooses the set $R_{c_2}$ based on $c_2= \argmax_{n\in \mathcal{I}} \sum_{m \in R_n^{1}} \lambda_m$. The same procedure is repeated until $C$ items are collected or all the items are chosen.

\textbf{LRU with aggregate requests.} Under SIM-LRU  an item is refreshed by the requests for all its neighbors. A naive approach to study a SIM-LRU cache is then to consider that it operates as a LRU cache with  request rates for each item equivalent to the sum of the request rates for all items in its neighborhood. One can then use the TTL approximation for LRU, leading to the following formulas:

    \begin{equation}\label{e:hit-rate-naive-che}
        \hit_{n} = 1- \e^{-\sum_{i\in \Nclosed{n}} \lambda_i  t_C } , \quad \newoccupancy_{n} =\hit_{n}.
    \end{equation}

We refer to the TTL approximation for LRU as `LRU', the greedy algorithm as `Greedy', and LRU with aggregate requests as `LRU-agg'.

\begin{figure*}
    \begin{subfigure}{\linewidth}
  \centering
  \includegraphics[width=0.7\linewidth]{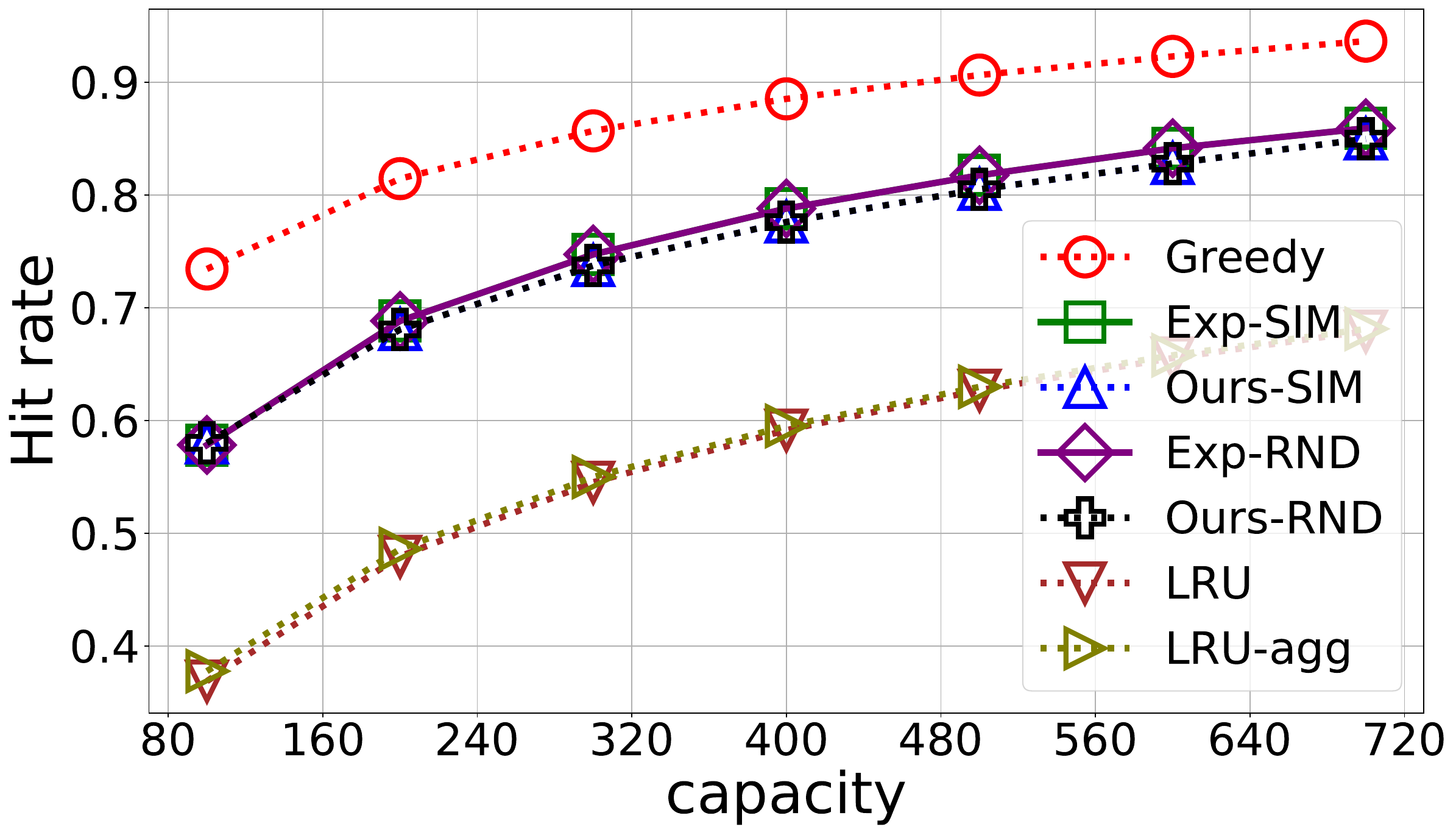}
\caption{Synthetic trace,  $\alpha = 2.5$, $d=1$, $25$ iterations}
\label{fig:hit-rate-d1}
\end{subfigure}
\hfill
\begin{subfigure}{\linewidth}
  \centering
  \includegraphics[width=0.7\linewidth, keepaspectratio]{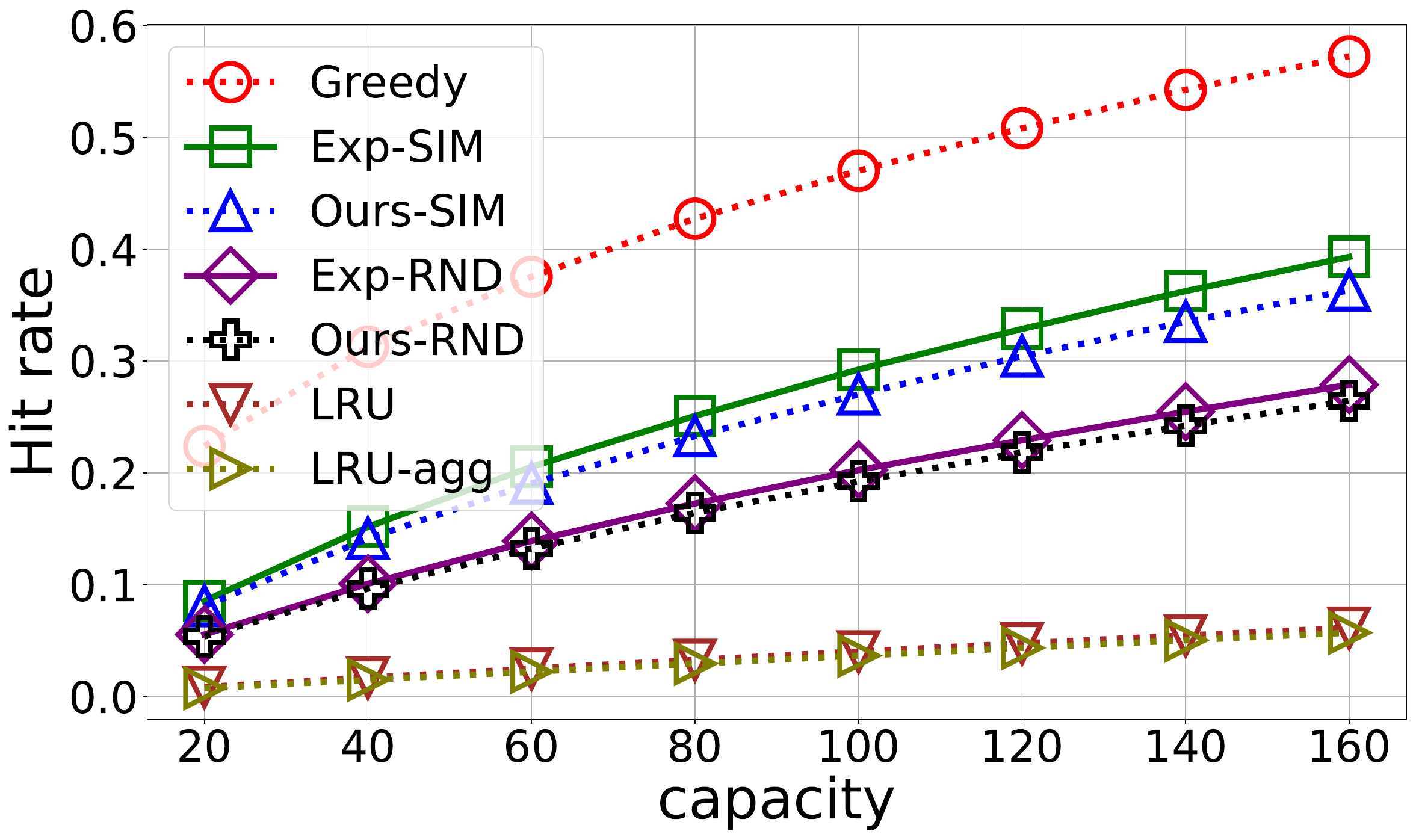}
  \caption{Synthetic trace, $\alpha=1.4$, $d=2$, $15$ iterations}
  \label{fig:hit-rate-d2}
\end{subfigure}
\hfill
 \begin{subfigure}{\linewidth}
      \centering  
    \includegraphics[width=0.7\linewidth,keepaspectratio]{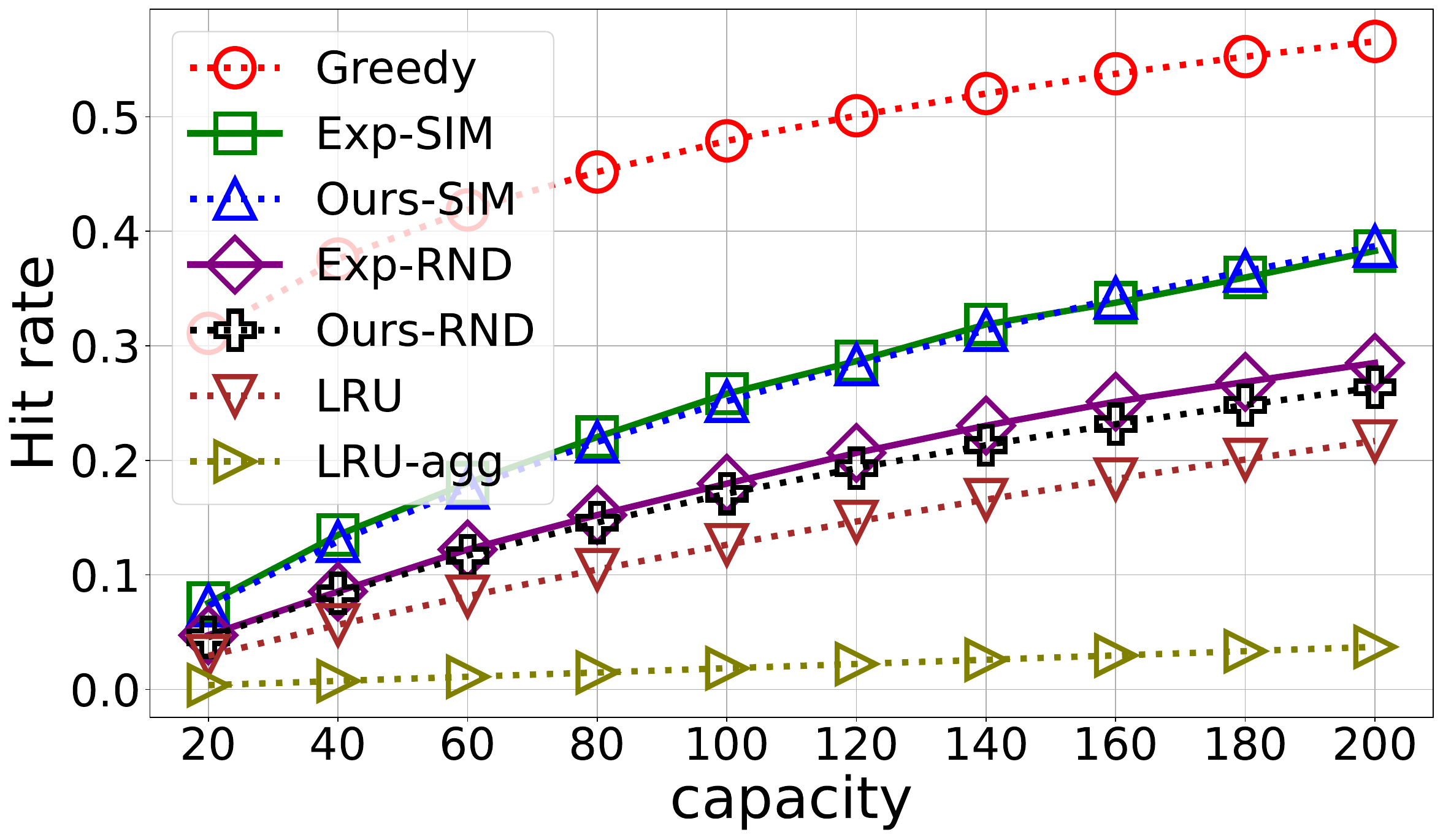}
    \caption{Amazon trace, $d=300$, $40$ iterations}
    \label{fig:hit-rate-amazon}
\end{subfigure}
    \caption{Hit rate versus cache capacity, $\beta=0.5$.}
    \label{fig:hit-rates} 
    \vspace{-0.1in}
\end{figure*}

\begin{figure}[t]
\begin{center}
    \begin{subfigure}{0.48\linewidth}
  \centering
  \includegraphics[width=0.99\linewidth]{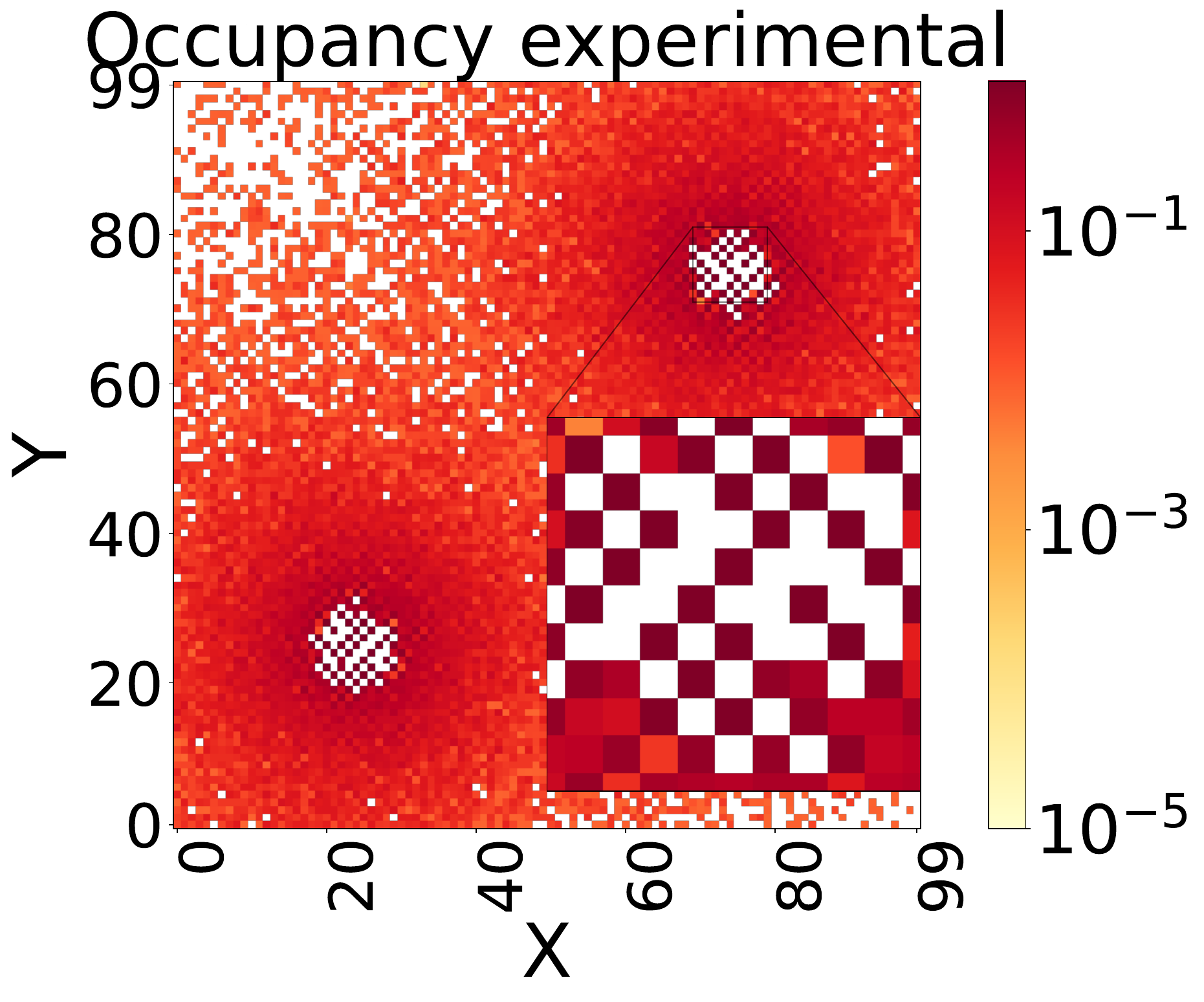}
\caption{$r=2\cdot 10^{5}$}
\label{fig:occupancy-exp-c500-d1}
\end{subfigure}
\end{center}
\begin{center}
\begin{subfigure}{0.48\linewidth}
  \centering
  \includegraphics[width=0.99\linewidth,keepaspectratio]{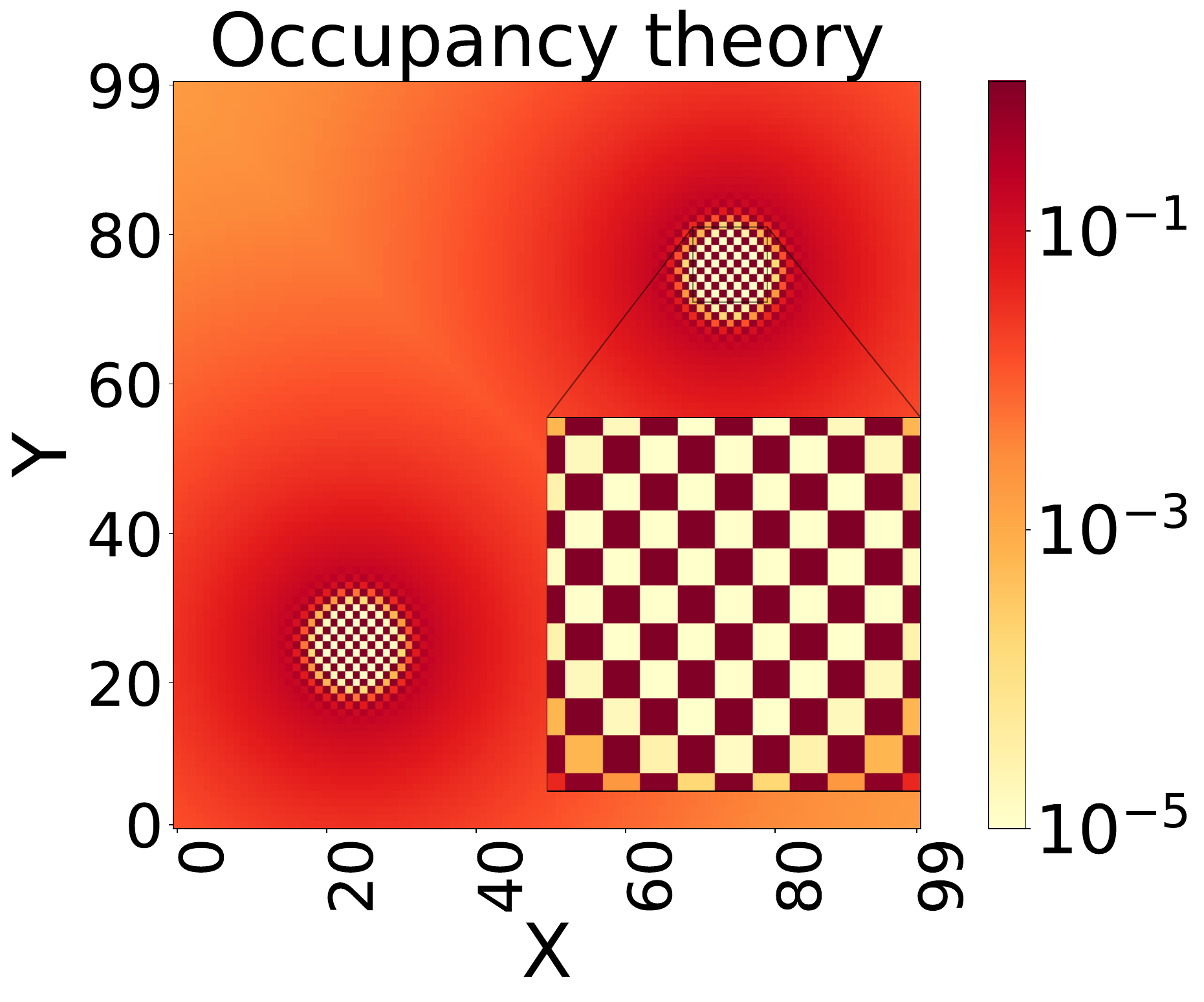}
\caption{$25$ iterations}
\label{fig:occupancy-theory-c500-d1}
\end{subfigure}
\end{center}
    \caption{Synthetic trace occupancies: $C=500$, $d=1$, $\alpha=2.5$.}
    \label{fig:occupancies}
\end{figure}

\subsection{RND-TTL approximation evaluation}
\label{ss:RND-TTL approximation-experiments}


We empirically compute the hit rate of similarity cache mechanisms using SIM-LRU and RND-LRU on both synthetic and real-world traces described in Section \ref{ss:Experimental-Setting}. In the case of synthetic traces, SIM-LRU and RND-LRU are utilized with similarity threshold parameter $d=1$ and $d=2$, for request process skewness $\alpha=2.5$ and $\alpha=1.4$, respectively. Additionally, given two distinct items $n$ and $m$, we set RND-LRU parameters $q_n(m)$ to $(\mathrm{dis}(n,m))^{-2}$. Note that when $d=1$, RND-LRU reduces to SIM-LRU. Results for the hit rate are averaged over the $50$ request processes for $\alpha=2.5$ and $\alpha=1.4$. The $95\%$ confidence intervals were smaller than $1.2 \cdot 10^{-3}$ in all the considered synthetic experiments for the hit rate computation. For the Amazon trace, SIM-LRU and RND-LRU are employed with a similarity threshold $d=300$. Furthermore, we choose $q_n(m)=(\mathrm{dis}(n,m))^{-0.2}$ as the RND-LRU parameters.
In all experiments, we refer to the empirical hit rates for SIM-LRU and RND-LRU as `Exp-SIM' and `Exp-RND', respectively.

For all the theoretical computations of the hit rate, the arrival rates $\boldsymbol{\lambda}$ for items are taken equal to the corresponding request  probabilities.

Our approach utilizes Algorithm~\ref{alg:cap} with parameter $\beta=0.5$ and a stopping condition determined by a fixed number of iterations. This algorithm is employed to estimate the approximate hit probabilities for each item, $\mathbf{h}$, and subsequently determines the overall cache hit rate $H$.

We refer to the latter estimate, for SIM-LRU and RND-LRU, as `Ours-SIM' and `Ours-RND', respectively. Possible alternative methods to estimate the hit rate are presented in Sec.~\ref{ss:alternative solutions}. The numerical values used for all the experiments are summarized in Table~\ref{tab:parameters-experiments}.

In Fig.~\ref{fig:hit-rates}, we show the empirical hit rate along with its estimates obtained through different approaches, for the two synthetic settings and for the Amazon trace. In  the considered settings, `Greedy' overestimates the hit rate. `LRU' and `LRU-agg', in contrast, underestimate it. `Ours-SIM' and `Ours-RND'  clearly outperform all the  alternative approaches presented in Sec.~\ref{ss:alternative solutions} in estimating the empirical hit rate, while tending to  underestimate it.
As `LRU' does not take into account the similarity between items, the gap between `LRU' and `Exp-SIM' reveals the benefits of similarity caching over exact caching.

For the synthetic settings in Figs.~\ref{fig:hit-rate-d1} and~\ref{fig:hit-rate-d2}, `LRU' and `LRU-agg' achieve similar hit rates. When $\boldsymbol{\tilde{\lambda}}= (\tilde{\lambda}_{n})_{n\in \mathcal{I}}$, where $\tilde{\lambda}_{n}=\sum_{m\in \Nclosed{n}}\lambda_m$,  is proportional to $\boldsymbol{\lambda}$, then `LRU' and `LRU-agg' have the same hit rate. In the choice of the popularity distribution in Figs.~\ref{fig:hit-rate-d1} and~\ref{fig:hit-rate-d2} (see \eqref{e:arrival-rates-square-two-hot-regions}), a popular item and its neighbors share similar rates, i.e., $\lambda_n \approx \lambda_m$ for $m\in \Nclosed{n}$. It follows that in the settings of Figs.~\ref{fig:hit-rate-d1} and~\ref{fig:hit-rate-d2}, the approximations~$\tilde{\lambda}_n\approx 5 \cdot \lambda_n$ and~$\tilde{\lambda}_n\approx 13 \cdot \lambda_n$ hold for the respective scenarios, especially for the popular items.  This provides insight into the comparable hit rates observed between 'LRU' and 'LRU-agg' in Figs.~\ref{fig:hit-rate-d1} and~\ref{fig:hit-rate-d2}.






While our approach provides the best estimates, we can observe that it slightly underestimates the hit rate. In order to understand this effect, we show in Fig.~\ref{fig:occupancies} the empirically estimated occupancy vector and the one produced by Algorithm~\ref{alg:cap}.
The proposed algorithm broadly captures the empirical occupancy  patterns, but  with  subtleties regarding symmetries. In particular, the zoom on Fig.~\ref{fig:occupancy-theory-c500-d1} shows that our approach produces a regular chess board  pattern. Some items are predicted to stay almost all the time in the cache while their $4$ neighbors are predicted to spend virtually no time in it. 
The corresponding empirical occupancy on Fig.~\ref{fig:occupancy-exp-c500-d1} shows a less symmetric pattern, implying that in this setup SIM-LRU is able to satisfy a group of requests using a smaller number of cache slots when compared against what is predicted by our approach. This, in turn, partially explains why our approach underestimates the hit~rate.


\subsection{ Convergence of Algorithm~\ref{alg:cap}} \label{ss:Convergence-Algorithm1-experiments}

The RND-TTL approximation selects the parameters $\boldsymbol{\lambda^i}$, $\boldsymbol{\lambda^r}$, and $T$ for the RND-TTL cache in a way that ensures the occupancy vector, as described in \eqref{e:occupancy-simlru-che}, is a fixed point of the function $\mathbf{G}$ defined in \eqref{e:Fixed-Point-G}. Algorithm~\ref{alg:cap} employs an iterative procedure aimed at finding a fixed point of $\mathbf{G}$, thereby determining the appropriate values for the RND-TTL cache's parameters. 

 Figure~\ref{fig:characteristic time t_C} shows the evolution of characteristic time $t_C$ and hit rate $H$ over different iterations. We observe that estimates of $H$ and $t_C$ by our algorithm converge in few iterations (less than $50$), under all considered scenarios. Note that $t_C(0)$, the value of $t_C$ at iteration $0$, is also the value of $t_C$ for `LRU' (see \eqref{e:tC-Defintion}). In addition, across all experiments, $t_C$ for `Ours-SIM' using Algorithm~\ref{alg:cap} converges to a value larger than $t_C(0)$.
Indeed, under LRU, $t_C$ is bounded by the time required for $C$ distinct items to be requested. For SIM-LRU and RND-LRU, in contrast, after $C$ distinct items are requested, an item previously in cache can remain there, despite not serving any requests. This occurs due to approximate hits, explaining why $t_C$ is larger for `Ours-SIM' than `LRU'. 

\begin{figure}
\begin{center}
\begin{subfigure}{0.8\linewidth}
  \centering
  \includegraphics[width=0.99\linewidth]{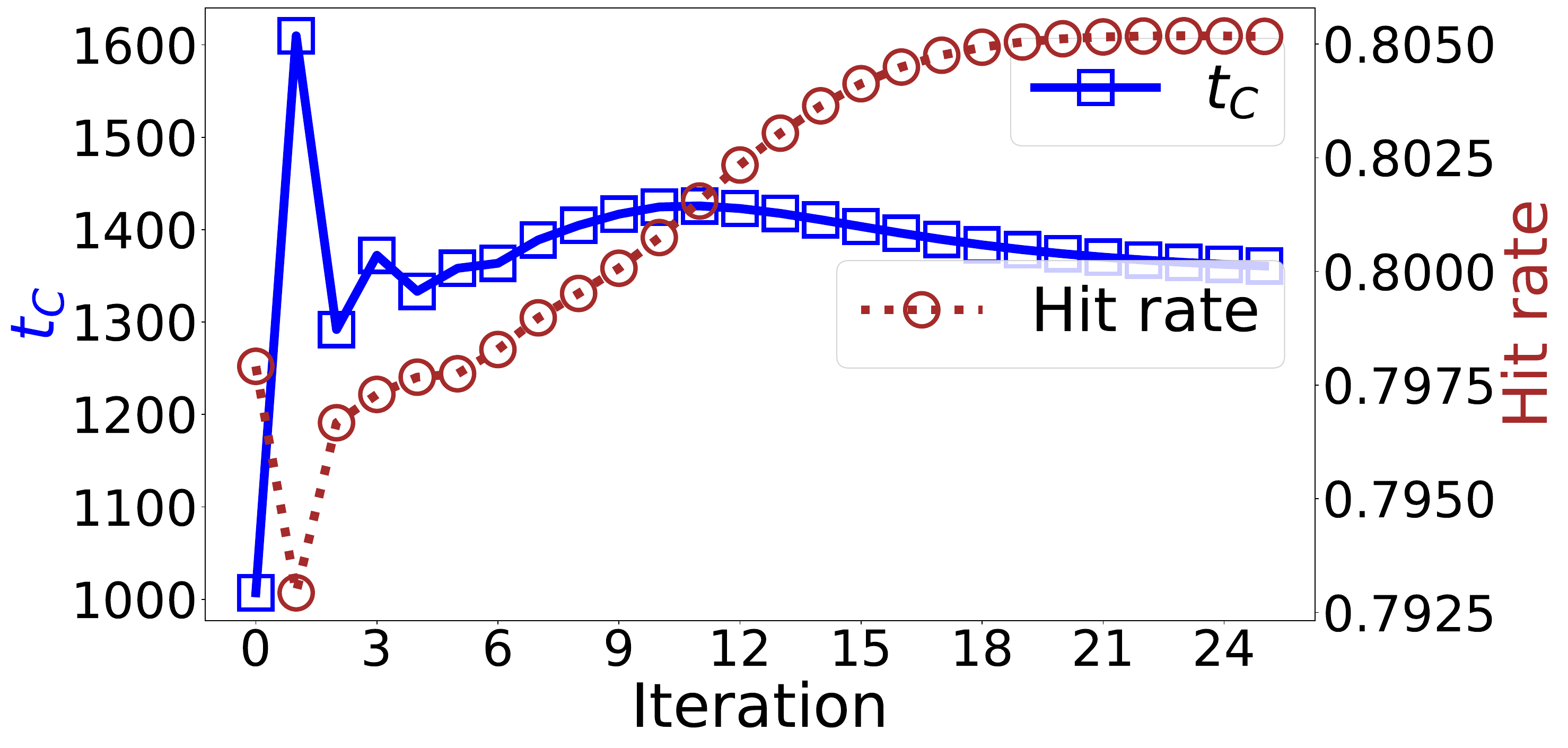}  
\caption{Synthetic trace, $ \alpha = 2.5$, $d=1$, $C=500$}
\label{fig:t_C-C500-d1} 
\end{subfigure}
\end{center}
 \begin{center}
 \begin{subfigure}{0.8\linewidth}
    
    \includegraphics[width=0.99\linewidth]{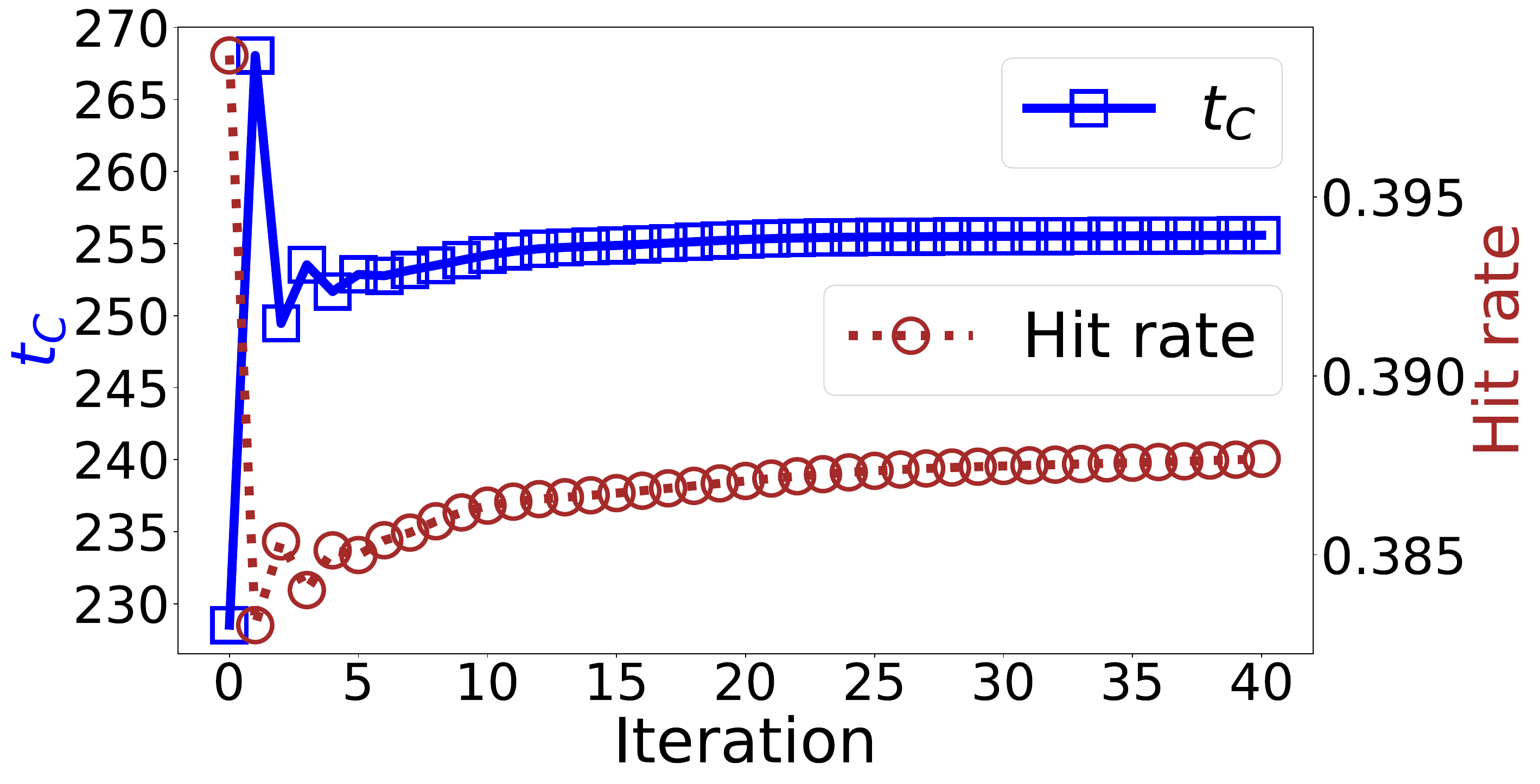}

    \caption{Amazon trace, $d=300$,  $C=200$ } 
    \label{fig:tc-amazon}
    \end{subfigure}
    \end{center}
    \caption{Characteristic time $t_C$ and hit rate in different iterations of Algorithm~\ref{alg:cap} for SIM-LRU.}
    \label{fig:characteristic time t_C}
\end{figure}

\begin{figure*}
    \begin{subfigure}{\linewidth}
  \centering
  \includegraphics[width=0.7\linewidth]{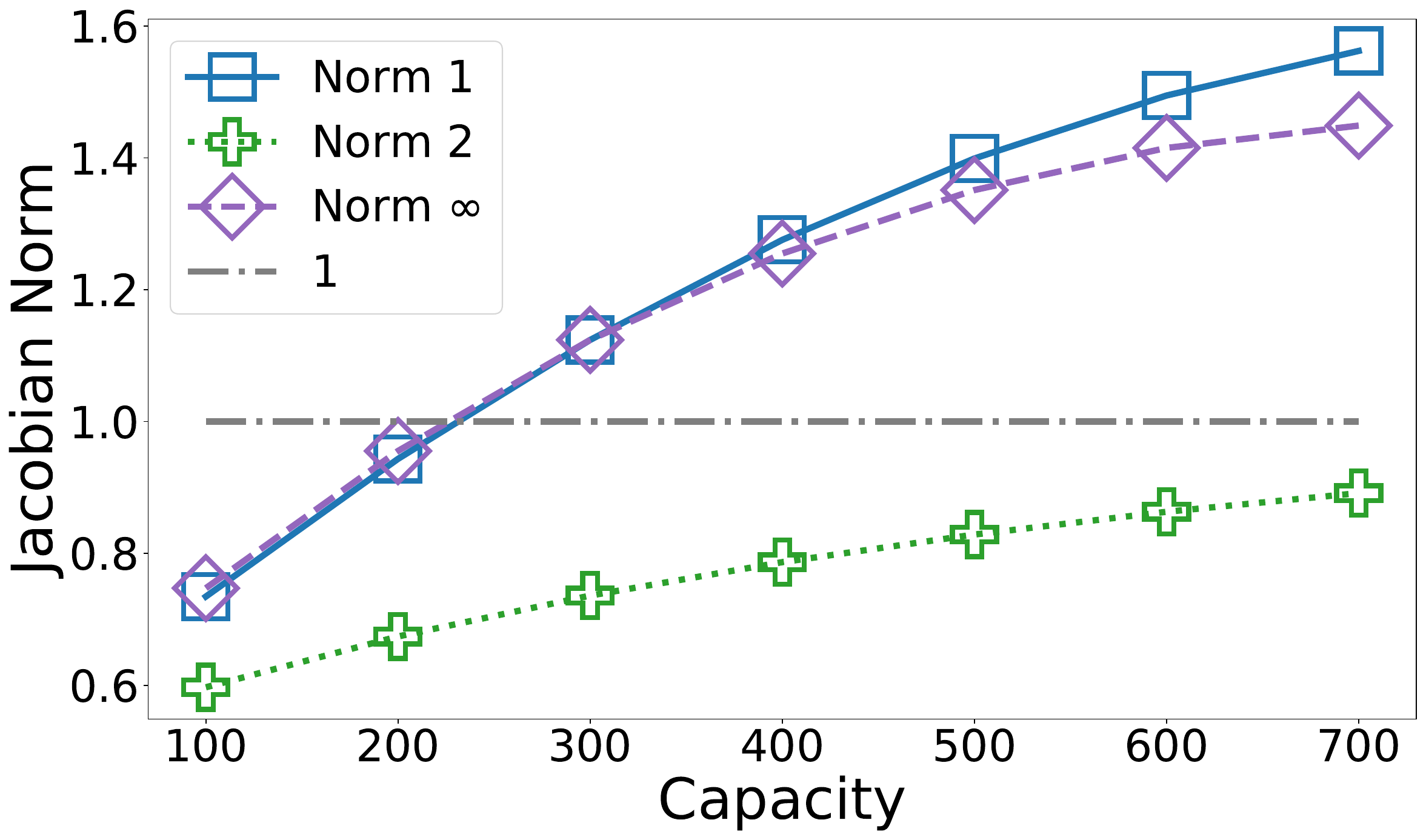}
\caption{Synthetic trace, $\alpha=2.5$,  $d=1$}
\label{fig:Jac-capacity-InitLRU-D1-B05}
\end{subfigure}
\hfill
\begin{subfigure}{\linewidth}
  \centering
  \includegraphics[width=0.7\linewidth, keepaspectratio]{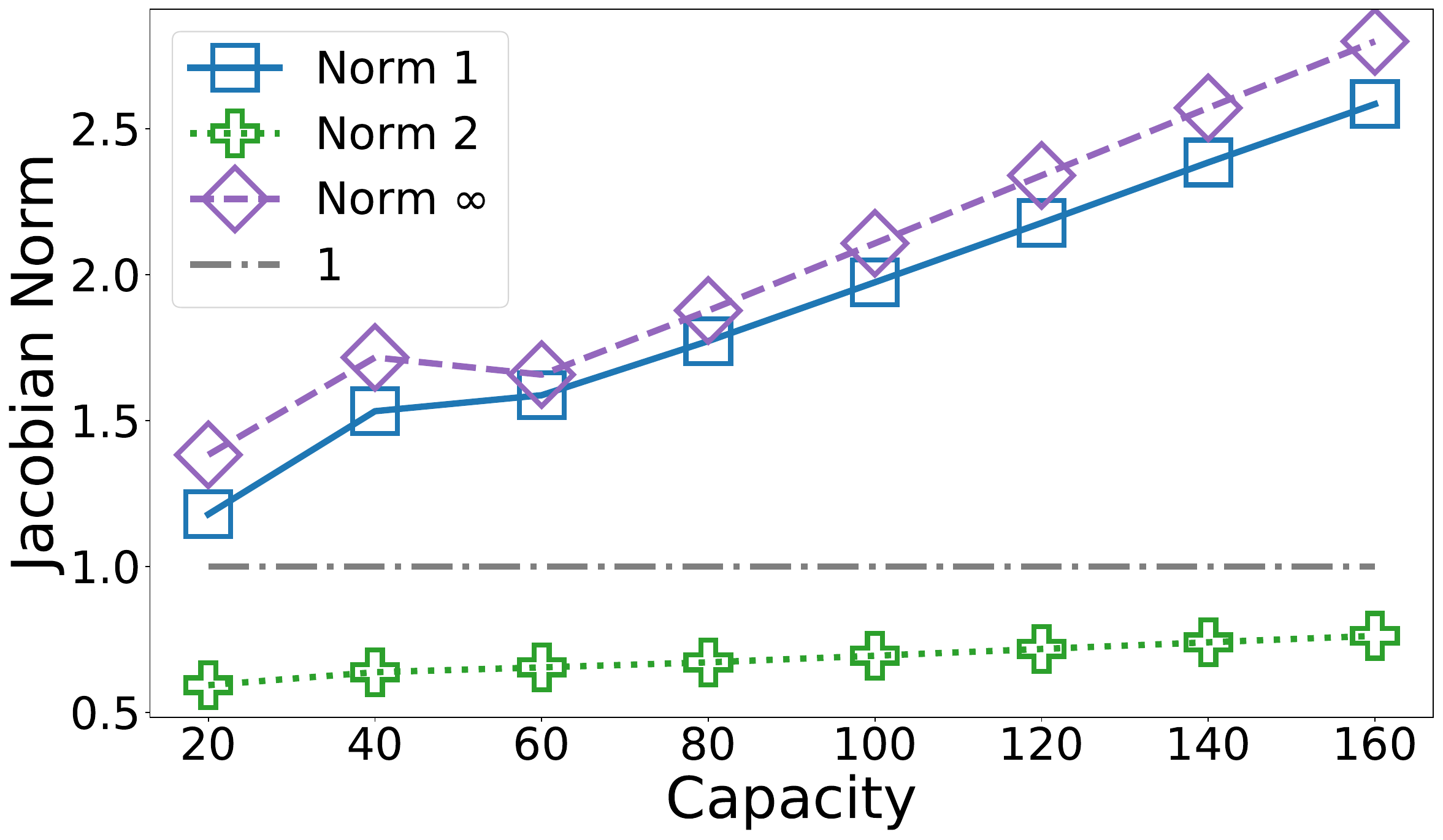}
  \caption{Synthetic trace, $\alpha=1.4$, $d=2$ }
  \label{fig:Jac-capacity-InitLRU-D2-B05}
\end{subfigure}
\hfill
 \begin{subfigure}{\linewidth}
      \centering  
    \includegraphics[width=0.7\linewidth,keepaspectratio]{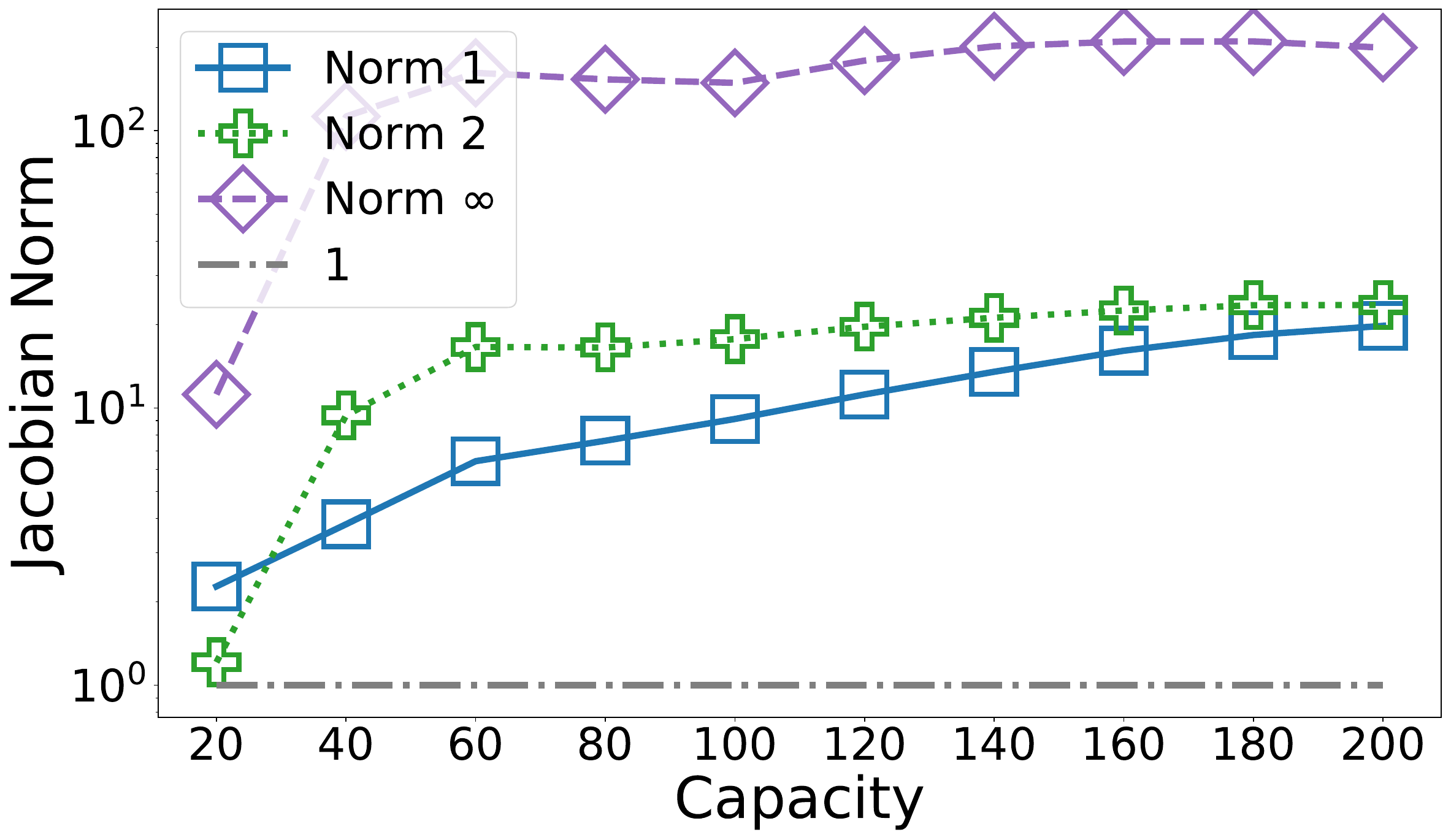}
    \caption{Amazon trace, $d=300$}
    \label{fig:Jac-Norm-amazon}
\end{subfigure}
    \caption{ Norm $\JacobianGB$ versus cache capacity,  $\beta=0.5$.}
    \label{fig:Jac-Norm-Capacity} 
    \vspace{-0.1in}
\end{figure*}

Proposition \ref{prop:jacobian-convergence} provides a sufficient condition for the convergence of Algorithm~\ref{alg:cap} with parameter $\beta$ towards a unique fixed point of $\mathbf{G}$. This condition requires an operator norm of the Jacobian matrix, $\JacobianGB(\mathbf{o})$ associated with the map $\mathbf{G}_{\beta}$, to be strictly less than $1$ for any $\mathbf{o}\in \Delta_C$.  

In Figure \ref{fig:Jac-Norm-Capacity}, we show the spectral norm and norms $1$ and infinity of $\JacobianGB(\mathbf{o}(0))$, where $\mathbf{o}(0)$ is given in line $1$ of Algorithm~\ref{alg:cap}, for the two synthetic traces and the Amazon trace described in Section \ref{ss:Experimental-Setting}. We show in \ref{app:Implementation-Details} details for the computation of $\JacobianGB$. In all the settings, $\beta$ is set to $0.5$. We observe in Figure \ref{fig:Jac-Norm-Capacity} that the norm of the Jacobian matrix $\JacobianGB$ increases with the cache capacity. For the synthetic traces, the spectral norm of $\JacobianGB$ is smaller than $1$ for all cache capacities, whereas in the Amazon trace, for all tested norms, $\norm{\JacobianGB}$ exceeds $1$. Nevertheless, Proposition \ref{prop:jacobian-convergence} provides only sufficient conditions and our results (see e.g. Fig. \ref{fig:tc-amazon}) suggest that our algorithm converges also on the Amazon trace.

\section{Conclusion}
\label{Conclusion}

We proposed a method, namely the RND-TTL approximation, for estimating the hit rate of a popular similarity caching policy, RND-LRU, under IRM. This method tunes the parameters of RND-TTL, a novel similarity cache model we introduced, and uses its hit rate as an estimation for RND-LRU's hit rate. The RND-TTL approximation involves solving a system of fixed point equations via a parameterized iterative algorithm. We studied the convergence of this algorithm and proposed a practical way for choosing its parameter.

Our experimental benchmark shows that our approach outperforms other methods one could use to predict the hit rate. In future work, we envision to investigate analytically the accuracy of our RND-TTL approximation, similarly to what was done in \cite{fricker2012versatile,jiang2018convergence} for classic caches.

\section*{Acknowledgement}
 This project was funded in part by CAPES,  CNPq, FAPERJ Grant JCNE/E-26/203.215/2017, in part by the French
Government through the “Plan de Relance” and “Programme
d’investissements d’avenir,” and by Inria under the exploratory action MAMMALS.

 \bibliographystyle{elsarticle-num-names} 
 \bibliography{references.bib}

 \appendix

\section{TTL based similarity caching algorithm}\label{app:algorithm-TTL}
A pseudo-code for TTL-based similarity caching is given in Algorithm \ref{alg:TTL-SIM2}. 
Note that in practice, we only need to keep track of TTL values greater than zero, corresponding to items that are in cache. However, to simplify presentation, we assume that all TTLs are stored at all time instants. In addition, we say that $m \in S_{\tau}$ if the TTL of $m$ at state $S_{\tau}$ is greater than zero. Finally, the algorithm assumes that the cache statically stores a  tombstone item  whose distance to all items is infinite.  Whenever a request arrives to an empty cache, the tombstone item is returned as the closest item in cache, in line~\ref{line:closet-item-in-cache1-2b}.

\begin{algorithm}
\caption{TTL-based similarity caching}\label{alg:TTL-SIM2} 
\begin{algorithmic}[1]
\STATE \textbf{Input:}
\STATE A sequence of requests $\mathbf{R}_T = (r_1, \ldots, r_T)$ and a sequence of time instants $(\tau_1, \ldots, \tau_T)$ when the requests occurred 
\STATE Threshold similarity $d$ and probabilities $(q_{n}(m))_{n,m \in \mathcal{I}^2}$
\STATE Initial TTL vector $\mathbf{S}_0 = (s_{0,1}, \ldots, s_{0,|\mathcal{I}|})$ where $s_{0,i}$ is the initial value of TTL of item $i$. 
\STATE \textbf{Output:}
\STATE The cache state at each   step $t$, $t=1, \ldots, T$.
\STATE \textbf{Algorithm:}
\FOR{$t = 1$ to $T$}
\STATE  $S_t \leftarrow S_{t-1}$
\STATE For all timers in $S_t$, decrement by $\tau_t - \tau_{t-1}$
 \IF{There are timers that reached value zero or negative}
 \STATE Evict corresponding items from cache
\ENDIF
    \STATE Compute the closest item to $r_t$ whose TTL is greater than zero in $S_{t}$ as $\hat{r}_t = \arg\min_{m \in S_{t}} \text{dis}(r_t, m)$. \label{line:closet-item-in-cache1-2b}
    \STATE Generate a uniform random number $u \in [0,1]$
    \IF{$u < q_{\hat{r}_t}(r_t)$  and $\text{dis}(\hat{r}_t, r_t) < d$} \label{line:condition-approximate hit1-2b}
        \STATE \textbf{Case 1:} Hit, encompassing exact/approximate hits
        \STATE Update $S_t$ to reset  timer of  $\hat{r}_t$ to $s_{0,\hat{r}_t}$ \label{line:action-approximate-hit1-2b}
    \ELSE
        \STATE \textbf{Case 2:} Miss
        \STATE   Update $S_t$ to add  $r_t$ to the cache and reset its timer to $s_{0,r_t}$\label{line:action-miss1-2b}
    \ENDIF
 \ENDFOR
\end{algorithmic}
 \end{algorithm}

\section{Proof of Proposition~\ref{prop:occupancy-RND-TTL}}

\label{app:proofoccupancy}

    To derive the occupancy of an item $n$, we first observe that the instants when item $n$ is evicted from the cache are regeneration points of a renewal process~\cite{ross2014introduction}. A renewal cycle consists of two consecutive time periods: a time period of duration $\Toff$, that starts immediately after item $n$ is evicted from the cache and ends when it re-enters the cache, and a time period of duration $\Ton$, that ends when item $n$ is evicted again from the cache.
From~\cite[Thm. 3.6.1, Example 3.6(A)]{ross1995stochastic}, the occupancy can be computed as: 
\begin{equation}
\label{e:occupancy}
     \newoccupancy_n = \frac{\E{\Ton}}{\E{\Toff} + \E{\Ton}}. 
\end{equation}

\noindent  We have that $\Ton$ verifies:   
\begin{equation}
\label{e:Ton-sum-value}
    \Ton = \sum_{j=1}^F Y_j + T,
\end{equation} 
where $(Y_j)_{j\in \{1, \ldots, F\}}$ are exponentially distributed random variables with parameter $\lambda_n^{r}$ such that $Y_j<T$  for  $j= 1, \ldots, F,$ and $F$ is a geometric random variable. Since we have: 
\begin{align}
   & \E{Y_j\given Y_j< T} = \frac{1}{\lambda_{n}^r} - \frac{T}{\exp(\lambda_{n}^{r} T) - 1}  ,\\
   & \E{F} = \exp(\lambda_{n}^{r} T)- 1,
\end{align}
we conclude from Wald's identity and~\eqref{e:Ton-sum-value} that:
\begin{equation}
    \label{e:ton-expectation}
        \E{T_n^{\mathrm{On}}} =  \frac{\e^{\lambda_{n}^{r} T}-1}{\lambda_{n}^{r}}.
\end{equation}
\noindent By combining \eqref{e:ton-expectation} and \eqref{e:occupancy} and observing that $\E{\Toff} = 1/ \lambda_{n}^i$, we get our result.

\section{Generalized Poisson Arrivals See Time Averages (PASTA) property}
\label{app:PASTA}
We derive in this appendix a generalization of the PASTA property that will be used in the proofs of Propositions \ref{prop:hit-rate-RND-TTL}--\ref{prop:refresh-rate-RND-LRU}. We will construct a stochastic process $Q$ that is not Poisson but whose jumps coincide with one of many Poisson processes (to be defined). In our generalization, we prove that the arrivals of this process $Q$ see time averages.

Let $(M(t))_{t\geq 0}$ be a stochastic process with finite state space $\Omega$. We assume that for every $S\in \Omega$,  $\lim_{t\to +\infty} \Proba{M(t)=S}$ exists and we denote it as $\pi_S$. We also assume that $ \lim_{T\to +\infty} \frac{1}{T} \int_{0}^T \mathds{1}(M(u)=S) du$ exists and is equal to $\pi_S$. 

For every $S\in \Omega$, we define a Poisson process $P_S$ with rate $\lambda_S$. For any $S$, the processes $\{ P_S(t+u) - P_S(t), \; u\geq 0 \}$ and $\{M(v), 0\leq v\leq t\}$ are assumed to be independent. This assumption is known as the lack of anticipation assumption \cite{wolff1982poisson,van1988conditional}. We construct a stochastic process $Q$ such that its jumps coincide with the jumps of $P_S$ when $M$ is in state $S$. 

Let $B$ be a subset of $\Omega$. We define the stochastic process $Y_B$ as the number of times $Q$ finds $M$ in $B$ normalized by $t$; we can write
\begin{align}
     Y_{B}(t) &\triangleq  \frac{1}{t}  \int_{0}^{t} \mathds{1}\left( M(u) \in  B \right) \, dQ(u) \; . 
\end{align}

\begin{theorem}[Generalized PASTA] 
\label{th:PASTA-MMPP}
  For every subset $B$ of $\Omega$, $Y_B(t)$ converges to $\sum_{S\in B} \lambda_S \pi_S$, as $t$ goes to $+\infty$, with probability $1$. 
\end{theorem}
\begin{proof}
    For every state $S\in \Omega$, we define the following two stochastic processes: 
    \begin{align}
      &V_S(t) \triangleq  \frac{1}{t} \int_{0}^{t} \mathds{1}\left( M(u) = S \right) \, du \; , \\ 
      &X_S(t) \triangleq \frac{1}{t}  \int_{0}^{t} \mathds{1}\left( M(u) = S \right) \, dQ(u) \; .
    \end{align}
    
\noindent $V_S(t)$ is the fraction of time spent by $M$ in state $S$ and $X_S(t)$ is the number of times $Q$ finds $M$ in state $S$ normalized by $t$. It is easy to check that: 
    \begin{align}
        \mathds{1}\left( M(u) = S \right) \, dQ(u) = \mathds{1}\left( M(u) = S \right) \, dP_S(u),
    \end{align}
since the jumps of $Q$ coincides with the jumps of $P_S$ when $M$ is in state $S$. The PASTA property \cite{wolff1982poisson} allows us to deduce that with probability $1$: 
\begin{align}
   \lim_{t\to +\infty} X_S(t)   = \lambda_S \pi_S~. 
\end{align}
Observing that $Y_B(t) = \sum_{S\in B} X_S(t)$, we deduce our result.  
\end{proof}
A similar result is proven in \cite[Sect. 3.3]{rosenkrantz1992some} for Markov-modulated Poisson processes. 

\begin{corollary}[Generalized PASTA]
\label{cor:PASTA-MMPP}
For any partition $(B_i)_{i\in \{1,\ldots ,l\} }$ of $B$ such that $\lambda_S = \lambda_i$ for any $S$ in $B_i$, we have that 
\begin{align}
    Y_B(t) \xrightarrow{ t\to+\infty} \sum_{i=1}^{l} \lambda_i \sum_{S\in B_i} \pi_S,  \quad \text{w.p. } 1. 
\end{align}
\end{corollary}
Observe that when $B=\Omega$, Theorem~\ref{th:PASTA-MMPP} and Corollary~\ref{cor:PASTA-MMPP} provide the rate of the process $Q$.

\section{Proof of Proposition \ref{prop:hit-rate-RND-TTL}}
\label{app:hit-rate-RND-TTL}
The proof uses the generalized PASTA property derived in~\ref{app:PASTA}, and we will redefine the relevant processes to serve our purpose. 

The stochastic process $M(t)$ refers to the set of cached items in the RND-TTL cache at time~$t$. For a given item $n$, we partition $M$'s state space $\Omega$ into $B_1=\{ S\in \Omega: \; n\in S \}$ and $B_2=\{ S\in \Omega: \; n\notin S \}$. For any cache state $S$ in $B_1$ we define $P_S$ as the request process, that is Poisson with rate $\lambda_n$ by Assumption~\ref{assum-IRM-request-process}. For any cache state $S$ in $B_2$ we define $P_S$ as the request process for $n$ thinned with probability $1-p_n^i$. We define now a stochastic process $Q_n$ whose jumps coincide with the jumps of $P_S$ when $M$ is in $S$. Corollary~\ref{cor:PASTA-MMPP} provides the rate of $Q_n$. So we only need to show that $Q_n(t)$ is nothing but the number of hits for item $n$ until time $t$ to find the hit rate and subsequently the hit probability.

In RND-TTL, whenever an item $n$ is in the cache, an exact hit occurs upon a request for $n$. In other words, if the cache state $M$ is in $B_1$, the number of hits for $n$ grows as its request process (and those added hits are all exact). Conversely, when $n$ is not in the cache, only an approximate hit may occur, with probability $1-p_n^i$. In other words, if the cache state $M$ is in $B_2$, the number of hits for $n$ grows as its request process thinned with probability $1-p_n^i$ (and those added hits are all approximate). 
It is clear then that the number of hits for $n$ until time $t$ is the process $Q_n(t)$ defined earlier. 

By Corollary~\ref{cor:PASTA-MMPP}, the rate of $Q_n$ (the hit rate) is $\lambda_n \cdot o_n + \lambda_n(1-p_n^i) (1-o_n)$, where $o_n=\sum_{S\in B_1} \pi_S$ and $1-o_n=\sum_{S\in B_2} \pi_S$. The hit probability readily follows concluding the proof.

\section{Proof of Proposition \ref{prop:entry-rate-RND-LRU}}
\label{app:proof-entry-rate-RND-LRU}

The proof is very similar to that for Proposition \ref{prop:hit-rate-RND-TTL} in~\ref{app:hit-rate-RND-TTL}. 

We redefine the stochastic process $M(t)$ to be the state of the RND-LRU cache at time~$t$.
For a given item $n$, we consider a subset $B$ of $M$'s state space $\Omega$ with all states that do not have $n$, and partition it into $B_{n}=\{ S\in \Omega: \; S\cap \Nclosed{n}=\emptyset \}$ and $B_{n,m}=\{ S\in \Omega: \; m\in S , S \cap \Nclosedin{m}{n}=\emptyset\}$ for every $m\in \Nopen{n}$. For any cache state $S$ in $B_{n}$, we define $P_S$ as the request process for $n$ that is Poisson with rate $\lambda_n$ by Assumption \ref{assum-IRM-request-process}. For any cache state $S$ in $B_{n,m}$ with $m\in \Nopen{n}$, we define $P_S$ as the request process for $n$ thinned with probability $1-q_{m}(n)$. The stochastic process $Q_n$ keeps the same definition: its jumps coincide with the jumps of (the redefined) $P_S$ when $M$ is in $S$. We stress that there are no jumps of $Q_n$ when $M$ is not in $B$ (this corresponds to defining a Poisson process with rate 0 for state $S\notin B$). Corollary~\ref{cor:PASTA-MMPP} provides then the rate of $Q_n$. Therefore we only need to show that $Q_n(t)$ is the number of insertions for item $n$ until time $t$ to find the insertion rate. 

In RND-LRU, when neither $n$ nor any of its neighbors are in the cache, a request for $n$ will correspond to a miss and thereby $n$ is inserted in the cache. In other words, when $M$ is in $B_{n}$, the number of insertions of $n$ grows as its request process. However, when $n$ is not in the cache but a neighbor is, $n$ may be inserted in the cache with some probability. Specifically, if $m$ is the closest neighbor of $n$ in the cache, $n$ will be inserted upon a request with probability $1-q_{m}(n)$. That is, when $M$ is in $B_{n,m}$, the number of insertions of $n$ grows as its request process thinned with probability $1-q_{m}(n)$.   

It is clear then that the number of insertions for $n$ until time $t$ is the process $Q_n(t)$ defined earlier. By Corollary~\ref{cor:PASTA-MMPP}, we find the following expression for the rate of $Q_n$ (the insertion rate)
    \begin{align}
        \tilde{\lambda}_{n}^i = \lambda_n \left(\Proba{B_{n}} + \sum_{m\in \Nopen{n}} (1-q_{m}(n)) \Proba{B_{n,m}}\right)~.
    \end{align}



\section{Proof of Proposition \ref{prop:refresh-rate-RND-LRU}}
\label{app:proof-refresh-rate-RND-LRU}

As in~\ref{app:proof-entry-rate-RND-LRU}, the process $M(t)$ refers to the state of the RND-LRU cache at time $t$. For a given item $n$, and for every $m\in \Nclosed{n}$, we define the set $B_{m,n}=\{ S\in \Omega: \; n\in S, S \cap \Nclosedin{n}{m} =\emptyset\}$.

For any cache state $S$ in $\Omega$, we consider the aggregation of thinned request processes for items in the subset $\{m\in \Nclosed{n}, S\in B_{m,n}\}$. For every $m$ in the aforementioned subset, the thinning probability of $m$'s request process is $q_{n}(m)$. Under Assumption \ref{assum-IRM-request-process}, this aggregated process is Poisson and we redefine $P_S$ to represent it; the rate of $P_S$ is then given by 

\begin{equation}
\label{e:rate-PS}
    \lambda_S = \sum_{m\in \Nclosed{n}} q_{n}(m) \lambda_m \mathds{1}(S\in B_{m,n})  .
\end{equation}
As in~\ref{app:hit-rate-RND-TTL} and~\ref{app:proof-entry-rate-RND-LRU}, the jumps of the stochastic process $Q_n$ coincide with the jumps of (the redefined) $P_S$ when $M$ is in $S$. By applying Theorem~\ref{th:PASTA-MMPP}, the rate of $Q_n$ will be readily available. It suffices then to show that $Q_n(t)$ is the number of times $n$'s timer is refreshed until time $t$ to find the refresh rate $\tilde{\lambda}_{n}^r$.

In RND-LRU, whenever an item $n$ is in the cache, a hit on $n$ occurs upon a request for $m$, with probability $q_n(m)$, if $n$ is the item in cache closest to $m$, and as a result $n$'s timer is refreshed. This holds for any neighbor $m$ of item $n$ (including $n$). In other words, for any neighbor $m$ such that the cache state $M$ is in $B_{m,n}$, $m$'s request process thinned with probability $q_n(m)$ contributes to the growth of the number of timer refreshes for $n$. No other process contributes to the refresh counting process, thereby, when $M$ is in state $S$, the number of times $n$'s timer is refreshed grows as the Poisson process $P_S$ with rate $\lambda_S$ given in \eqref{e:rate-PS}.


It follows then that the refresh counting process for $n$ is nothing but the process $Q_n$ defined earlier. We apply Theorem~\ref{cor:PASTA-MMPP} (with $B=\Omega$) to derive the rate of $Q_n$ (the refresh rate). We use \eqref{e:rate-PS} to write

\begin{subequations}
\begin{eqnarray}
\tilde{\lambda}_{n}^r &=& \sum_{S\in \Omega} \sum_{m\in \Nclosed{n}} q_{n}(m) \lambda_m \mathds{1}(S\in B_{m,n}) \, \pi_S  \\
    &=& \sum_{m\in \Nclosed{n}} q_{n}(m) \lambda_m  \sum_{S\in \Omega}  \mathds{1}(S\in B_{m,n})  \pi_S \\
    &=& \sum_{m\in \Nclosed{n}} q_{n}(m) \lambda_m \sum_{S\in B_{m,n}} \pi_S \\
    &=& \sum_{m\in \Nclosed{n}} q_{n}(m) \lambda_m \,\Proba{B_{m,n}}.
\end{eqnarray}
\end{subequations}

\section{Proof of Lemma \ref{lem:tc-function}} \label{app: tC-Function}

    Let $\mathbf{o}\in \Delta_C$. Observe that $F(\mathbf{o},0)=-C<0$ and that $F(\mathbf{o},\cdot)$ is an increasing and continuous function in $\mathbb{R}^{+}$ since it is the sum of increasing and continuous functions in $\mathbb{R}^{+}$. Therefore, to show the existence of a root $T_0$ for $F(\mathbf{o},\cdot)$, it suffices to show that $\lim_{T\to +\infty} F(\mathbf{o}, T)>0$, thanks to the intermediate value theorem.

\subsection{Proof idea}
    Next, we show that $\lim_{T\to +\infty} F(\mathbf{o}, T)>0$. To this aim, we note under the cover  condition~\eqref{e:condition-existence-fixed-point} at least $C+1$ items are required to cover all  the catalog. Therefore, given  $\mathbf{o} \in \Delta_C$,  
    we have  a strictly positive insertion rate $E_n(\mathbf{o}) > 0$ for at least $C+1$ items.  
Now, let's consider a corresponding alternative TTL cache in which at least $C+1$ items have  a strictly positive arrival rate and let $T\to \infty$.  In this alternative   cache we have  at least  $C+1$  statically stored items in steady state,  implying that    $\lim_{T\to +\infty} F(\mathbf{o}, T)>0$ in the original system.  The details follow below.

 \subsection{Proof details}   
Let $\mathcal{M}_{\mathbf{o}}$ be the set of items with occupancy equal to $1$ in $\mathbf{o}$. More specifically,~$\mathcal{M}_{\mathbf{o}}=\{ n\in \mathcal{I}:\; o_n =1  \}$. We show that $\lim_{T\to +\infty} F(\mathbf{o}, T)>0$ as follows

     \begin{align} \label{e:Lemma-TC-2} 
        \mathbf{o} \in \Delta_C
&\implies \# \mathcal{M}_{\mathbf{o}} \leq C \\ \label{e:Lemma-TC-3} 
&\implies \# \bigcup_{m \in \mathcal{M}_{\mathbf{o}}} \Nopen{m} < N-C \\ \label{e:Lemma-TC-4}
&\implies \#  \left\{ n\in \mathcal{I}: \; \prod_{m\in \Nopen{n}} (1-o_m) = 0 \right\} <N-C \\  \label{e:Lemma-TC-5} 
&\implies \#  \left\{n\in \mathcal{I}: \; \prod_{m\in \Nopen{n}} (1-o_m) > 0 \right\}> C \\  \label{e:Lemma-TC-6} 
&\implies \#  \left\{ n\in \mathcal{I}: \; E_n(\mathbf{o}) > 0 \right\} > C \\   \label{e:Lemma-TC-7} 
&\implies \#  \left\{ n\in \mathcal{I}: \; \lim_{T\to +\infty} g\left(R_n(\mathbf{o}), E_n(\mathbf{o}), T\right)  =1 \right\} > C \\ \label{e:Lemma-TC-8} 
&\implies \lim_{T\to +\infty} F(\mathbf{o}, T) > 0
    \end{align}
The transition from~\eqref{e:Lemma-TC-2}~to~\eqref{e:Lemma-TC-3} is by applying the cover condition \eqref{e:condition-existence-fixed-point} for~$\mathcal{M_{\mathbf{o}}}$. 

Noting  that~$n\in \Nopen{m} \iff m\in \Nopen{n}$, the passage from~\eqref{e:Lemma-TC-3} to~\eqref{e:Lemma-TC-4} follows from the fact that, given $n$, $$\prod_{m\in \Nopen{n}} (1-o_m) = 0 \iff \exists m\in \mathcal{M}_{\mathbf{o}
}: \; m\in \Nopen{n} \iff \exists m\in \mathcal{M}_{\mathbf{o}
}: \; n\in \Nopen{m}.$$ 

The step from~\eqref{e:Lemma-TC-5} to~\eqref{e:Lemma-TC-6} is by definition of the function $\fentry$ (see \eqref{e:lambda-entry-independence}). We find \eqref{e:Lemma-TC-7} from \eqref{e:Lemma-TC-6} since we have
    \begin{align}\label{e:Lemma-TC-1} 
        \lim_{T\to +\infty} g\left(R_n(\mathbf{o}), E_n(\mathbf{o}), T\right) = \begin{cases}
                0 &\text{ if }  E_n(\mathbf{o}) = 0, \\
                1 &\text{ otherwise},
        \end{cases}
    \end{align}
by definition of $g$ (see \eqref{e:occupancy-g-function}). Finally, the transition from \eqref{e:Lemma-TC-7} to \eqref{e:Lemma-TC-8} follows from $\lim_{T\to +\infty} F(\mathbf{o}, T)= \sum_{n\in \mathcal{I}} \lim_{T\to +\infty} g\left(R_n(\mathbf{o}), E_n(\mathbf{o}), T\right)-C$. 

We deduce the existence of a root of $F(\mathbf{o}, \cdot)$, $T_0$. Observe that $F(\mathbf{o}, \cdot)$ is strictly increasing because there exist $m\in \mathcal{I}$ such that $g\left(R_m(\mathbf{o}), E_m(\mathbf{o}), \cdot\right)$ is strictly increasing (see \eqref{e:Lemma-TC-1} and~\eqref{e:occupancy-g-function}). It follows that $T_0$ is unique, which concludes the proof.

\section{Proof of Lemma \ref{lem:tc-class-C1}}\label{app:tC-class-C1}

Let $j\in \mathcal{I}$ and $\mathbf{o}\in \Delta_C$. We first show that $\ftc$ is differentiable. For this aim, we prove the existence of the partial derivatives of $\ftc$, i.e., the existence of 
\begin{align}
\lim_{\epsilon\to 0} \frac{\ftc(\mathbf{o} + \epsilon \mathbf{e}_j) - \ftc(\mathbf{o}) }{\epsilon},
\end{align}
where $\mathbf{e}_j$ is the canonical vector in $\mathbb{R}^{N}$ that has $1$ in the $j$-th component and $0$ in all the other components.

Let $A= (\mathbf{o}, \ftc(\mathbf{o}))$ and  $B_j= (\mathbf{o}+ \epsilon \mathbf{e}_j , \ftc(\mathbf{o}+\epsilon \mathbf{e}_j))$. Observe that $F$ is continuously differentiable over $\Delta_C \times \mathbb{R}^{+}$ since it is the sum of the continuously differentiable functions. The \textit{Mean Value Theorem} for $F$ shows the existence of $L_j$ in the segment between $A$ and $B_j$ ($L_j\in \{  (1-\lambda) A + \lambda B_j: \; \lambda \in [0,1]  \}$) such that
\begin{align}
    F(B_j)- F(A) =  \nabla F(L_j) \cdot (B_j-A).
\end{align}
We have that $F(A)=F(B_j)=0$ by definition of $\ftc$ (see \eqref{e:def-tC}). It follows that $\nabla F(L_j) \cdot (B_j-A) = 0 $ and we have 

\begin{align}
    \epsilon \frac{\partial F}{\partial o_j}(L_j) +\left(\ftc(\mathbf{o} + \epsilon \mathbf{e}_j) - \ftc(\mathbf{o})\right) \cdot  \frac{\partial F}{\partial T} (L_j)=0. 
\end{align}
Recall that $F(\mathbf{x},\cdot)$ is continuous and strictly increasing (see \ref{app: tC-Function}), i.e., $\frac{\partial F}{\partial T}>0$. Therefore
\begin{align}
    \frac{\ftc(\mathbf{o} + \epsilon \mathbf{e}_j) - \ftc(\mathbf{o})}{\epsilon} = - \frac{\partial F}{\partial o_j}(L_j) \cdot \left( \frac{\partial F}{\partial T} (L_j)\right)^{-1}. 
\end{align}
When $\epsilon \to 0$, $B_j$ converges to $A$. Thereby, the partial derivative of $\ftc$ with respect to $o_j$ exists and is continuous and is expressed in \eqref{e:gradient-tC}, which completes the proof.  

We stress that \cite{de2012implicit}[Th.1] provides a similar proof. 

\section{Proof of Proposition~\ref{prop:jacobian-convergence}}\label{app:proof-jac-GB}


The Jacobian matrix of a function $f: 
 \mathbb{R}^{n} \to \mathbb{R}^{m}$ is a rectangular matrix with $m$ rows and $n$ columns. The element in the $i$-th row and $j$-th column  represents the partial derivative of the $i$-th component of   function $f$ with respect to the $j$-th variable.
 

Let $\mathbf{o}_0\in \Delta_C$ and $P_i =(R_i(\mathbf{o}_0),  E_i (\mathbf{o}_0),   t_C(\mathbf{o}_0))$. We use the chain rule on $\mathbf{G}$ to compute the partial derivative of its $i$-th component, $G_i$, with respect to the $j$-th variable. 

\begin{align}
\frac{\partial G_i}{\partial o_j}(\mathbf{o}_0) =
\frac{\partial g}{\partial x_1}(P_i) \frac{\partial R_i}{\partial o_j}(\mathbf{o}_0) 
+ \frac{\partial g}{\partial x_2}(P_i) \frac{\partial E_i}{\partial o_j}(\mathbf{o}_0)  
+ \frac{\partial g}{\partial x_3}(P_i) \frac{\partial t_C}{\partial o_j}(\mathbf{o}_0).
\end{align} 
Observe that 

\begin{align}
\left(\frac{\partial g}{\partial x_1}(P_i) \frac{\partial R_i}{\partial o_j}(\mathbf{o}_0)\right)_{1\leq i\leq n \atop 1\leq j\leq n }  = \text{Diag}\left(\partial_1 \foccupancy \right)&\cdot \mathcal{J}_{\frefresh}(\mathbf{o}_0),  
\end{align}
and 
\begin{align}
\left(\frac{\partial g}{\partial x_2}(P_i) \frac{\partial E_i}{\partial o_j}(\mathbf{o}_0)\right)_{1\leq i\leq n \atop 1\leq j\leq n } =  \text{Diag}\left(\partial_2 \foccupancy \right) \cdot \mathcal{J}_{\fentry}(\mathbf{o}_0).
\end{align}
To prove \eqref{e:Jacobian-G-Beta}, it suffices to prove that
\begin{align} \label{e:lem-Jac-GB-0}
     \left( \frac{\partial g}{\partial x_3}(P_i) \frac{\partial t_C}{\partial o_j}(\mathbf{o}_0)\right)_{1\leq i\leq n \atop 1\leq j\leq n } =- \frac{\partial_{3} \foccupancy^\intercal }{\sum_{n\in \mathcal{I}} \partial_{3} \foccupancy_n}\cdot \left( \partial_1 \foccupancy \cdot \mathcal{J}_{\frefresh}(\mathbf{o}_0) + \partial_2 \foccupancy\cdot \mathcal{J}_{\fentry}(\mathbf{o}_0)  \right).
\end{align}
To this aim, we compute the partial derivatives of $t_C$. Let $Q = (\mathbf{o}_0, t_C(\mathbf{o}_0))$. Using Lemma \ref{lem:tc-class-C1}, we have 
\begin{align}\label{e:lem-Jac-GB-1}
    \frac{\partial t_C}{\partial o_j}(\mathbf{o}_0) = -  \frac{\partial F}{\partial o_j}(Q) \cdot \left(\frac{\partial F}{\partial T}(Q) \right)^{-1}. 
\end{align}
From the definition of $F$ (see \eqref{e:F-capacity-function}), we get 
\begin{align}\label{e:lem-Jac-GB-2}
    \frac{\partial F}{\partial T}(Q)  = \sum_{n\in \mathcal{I}} \frac{\partial g}{\partial x_3}(P_n),
\end{align}
and
\begin{align}\label{e:lem-Jac-GB-3}
    \frac{\partial F}{\partial o_j}(Q)  = \sum_{n\in \mathcal{I}} \frac{\partial g}{\partial x_1}(P_n) \frac{\partial R_n}{\partial o_j}(\mathbf{o}_0) + \sum_{n\in \mathcal{I}} \frac{\partial g}{\partial x_2}(P_n) \frac{\partial E_n}{\partial o_j}(\mathbf{o}_0).
\end{align}
Plugging \eqref{e:lem-Jac-GB-2} and \eqref{e:lem-Jac-GB-3} in \eqref{e:lem-Jac-GB-1} we deduce \eqref{e:lem-Jac-GB-0}. This concludes the proof.


\section{Proof of Proposition~\ref{lem:effect-beta-spectral-norm}} 

\label{app:properties}

Let $\beta\in [0,1]$ and $\mathbf{o} \in \Delta_C$. In order to show that $(a,b)\subset Y(\mathbf{o})$,  where $a$ is given in \eqref{e:lem-Yo-a} and $b$ is given in \eqref{e:lem-Yo-b}, we first show that $\norm{\JacobianGB(\mathbf{o})}_2^{2}< (1+ \gamma ) \cdot \beta^{2}  -  ( 2 \gamma - \eta )  \cdot \beta + \gamma$. Next we observe that if \eqref{e:condition-subsest-Y} is verified, the aforementioned quadratic function in $\beta$ is smaller than $1$ if and only if $\beta$ lies within the interval $(a,b)$, to conclude that $(a,b)\subset Y(\mathbf{o})$.

We show now that $\norm{\JacobianGB(\mathbf{o})}_2^{2}< (1+ \gamma ) \cdot \beta^{2}  -  ( 2 \gamma - \eta )  \cdot \beta + \gamma$. We denote the spectral radius of a matrix $M$ as $\rho(M)$. We denote  the $N$-dimensional identity matrix  as $I_N$ and let
$$A=\mathcal{J}_{\mathbf{G}}(\mathbf{o}).$$
We compute the square of the spectral norm of $\JacobianGB(\mathbf{o})$ as follows
\begin{subequations}
\begin{eqnarray} \label{e:proof-lem1-0}
\norm{\JacobianGB(\mathbf{o})}_2^{2} 
&=& \left(\norm{(1-\beta) A + \beta I_N}_{2} \right)^2\\
\label{e:proof-lem1-1}
&=& \rho\left( (1-\beta)^{2} A\cdot A^\intercal + \beta (A + A^\intercal) + \beta^2 I_{N} \right) \\ 
\label{e:proof-lem1-2}
&\leq& \rho\left( (1-\beta)^{2} A\cdot A^\intercal + \beta (A + A^\intercal) \right) + \beta^2\\ 
\label{e:proof-lem1-3}
&=& \norm{(1-\beta)^{2} A\cdot A^\intercal + \beta (A + A^\intercal)}_2 + \beta^2 \\  \label{e:proof-lem1-4}
&\leq& \norm{(1-\beta)^{2} A\cdot A^\intercal}_2 +  \norm{\beta (A+A^\intercal)}_2 +\beta^2\\ 
\label{e:proof-lem1-5}
&=& (1-\beta)^{2} \gamma +  \beta \eta + \beta^2 \\ 
\label{e:proof-lem1-6}
&=&  (1+ \gamma ) \cdot \beta^{2}  -  ( 2 \gamma - \eta )  \cdot \beta + \gamma.
\end{eqnarray}
\end{subequations}

The passage from \eqref{e:proof-lem1-0} to \eqref{e:proof-lem1-1} is by definition of the spectral norm, i.e., $\norm{M}_2 = \sqrt{\rho(M \cdot M^\intercal)}$. The passage from \eqref{e:proof-lem1-1} to \eqref{e:proof-lem1-2} is by observing that $\rho(M + \beta^2 I_n) \leq \rho(M) + \beta^2$ for any matrix $M$. The passage from \eqref{e:proof-lem1-2} to \eqref{e:proof-lem1-3} follows from the fact that for a symmetric matrix $M$, its spectral norm and its spectral radius coincides, i.e.,  $\rho(M) =\norm{M}_2$. The passage from \eqref{e:proof-lem1-3} to \eqref{e:proof-lem1-4} is via the triangular inequality. This concludes the proof.


\section{Implementation Details}
\label{app:Implementation-Details}

 When computing $\norm{\JacobianGB(\mathbf{o})}$, we follow a specific procedure. First, we use the formula in \eqref{e:jacobian-G-Experession} to compute the Jacobian matrix $\mathcal{J}_{\mathbf{G}}$. To compute the Jacobian matrices $\mathcal{J}_{\mathbf{E}}$ and $\mathcal{J}_{\mathbf{R}}$, we utilize the "torch.autograd.functional.jacobian" function from the PyTorch library \cite{NEURIPS2019_9015}. However, the vectors $\partial \mathbf{g}_1$, $\partial \mathbf{g}_2$, and $\partial \mathbf{g}_3$ are implemented separately without relying on the aforementioned PyTorch function. This is done to ensure accurate results by avoiding potential errors that may arise from floating point precision.

\end{document}